\documentclass[preprint,12pt]{elsarticle}
\makeatletter
\def\ps@pprintTitle{%
  \let\@oddhead\@empty
  \let\@evenhead\@empty
  \let\@oddfoot\@empty
  \let\@evenfoot\@oddfoot
}
\makeatother

\usepackage[ruled,vlined]{algorithm2e}

\usepackage{amsmath}
\usepackage{tikz} 
\usepackage{array}
\usepackage{xcolor}

\usepackage{amsmath,amssymb,amsfonts}
\usepackage{color,soul}
\usepackage{float}
\usepackage{algorithmic}
\usepackage{graphicx}
\usepackage{textcomp}
\def\BibTeX{{\rm B\kern-.05em{\sc i\kern-.025em b}\kern-.08em
    T\kern-.1667em\lower.7ex\hbox{E}\kern-.125emX}}

\usepackage{amsthm,amscd,bm}

\newcommand{\al}{G}

\newcommand{\N}{\ensuremath{\mathbb{N}}\xspace}

\newcommand{\entropy}[1]{\mathcal{#1}}
\newcommand{\cyl}[1]{{Cyl(#1)}}

\newcommand{\Z}{\mathrm{Z}}
\newcommand{\ZZ}{\mathbb{Z}}

\newcommand{\Hom}{\mathrm{Hom}}
\newcommand{\Aut}{\mathrm{Aut}}
\newcommand{\Ker}{\mathrm{Ker}}
\newcommand{\End}{\mathrm{End}}

\newcommand{\Imma}{\mathrm{Im}}

\newcommand{\locrule}{\ensuremath{f}}
\newcommand{\glorule}{\ensuremath{{\cal F}}}

\usepackage{xspace, comment, mathdots}

\newcommand{\ie}{i.e.\@\xspace}

\newcommand{\gzd}{{\al}^{\ZZ}}

\newtheorem{theorem}{Theorem}
\newtheorem{definition}{Definition}
\newtheorem{remark}{Remark}
\newtheorem{example}{Example}
\newtheorem{lemma}{Lemma}

\newtheorem{corollary}{Corollary}

\newtheorem{proposition}{Proposition}

\newtheorem{conjecture}{Conjecture}

\newtheorem{question}{Question}

\journal{Journal of Computer and System Sciences}

\begin{document}

\begin{frontmatter}

\title{A Divide and Conquer Algorithm for Deciding Group Cellular Automata Dynamics}

\author[liceo]{Niccol\`o Castronuovo}
\ead{niccolo.castronuovo@studio.unibo.it}

\author[unimib]{Alberto Dennunzio\corref{cor}}
\ead{alberto.dennunzio@unimib.it}

%\author[UCA]{Enrico Formenti}
%\ead{enrico.formenti@univ-cotedazur.fr}

%\author[mfo]{Darij Grinberg}
%\ead{darijgrinberg@gmail.com}

\author[unibo]{Luciano Margara}
\ead{luciano.margara@unibo.it}
\cortext[cor]{Corresponding author}

\affiliation[liceo]{organization={Liceo ``A. Einstein''},%Department and Organization
            %addressline={ciao}, 
            city={Rimini},
            postcode={47923}, 
            %state={Rimini},
            country={Italy}}

%\address[liceo]{Liceo ``A. Einstein'', 47923 Rimini, Italy}

\affiliation[unimib]{organization={Dipartimento di Informatica, Sistemistica e Comunicazione,
  Università degli Studi di Milano-Bicocca},%Department and Organization
            addressline={Viale Sarca 336/14}, 
            city={Milano},
            postcode={20126}, 
            %state={Rimini},
            country={Italy}}
            
%\address[unimib]{Dipartimento di Informatica, Sistemistica e Comunicazione,
  %Università degli Studi di Milano-Bicocca,
  %Viale Sarca 336/14, 20126 Milano, Italy}
  
%\address[UCA]{Universit\'e C\^ote d'Azur, CNRS, I3S, France}

%\address[mfo]{Mathematisches Forschungsinstitut Oberwolfach, Schwarzwaldstr. 9-11, 77709 Oberwolfach-Walke, Germany}

\affiliation[unibo]{organization={Department of Computer Science and Engineering, University of Bologna, Cesena Campus},
            addressline={Via dell'Universita 50}, 
            city={Cesena},
            postcode={47521}, 
            %state={Rimini},
            country={Italy}}
            
%\address[unibo]{Department of Computer Science and Engineering, University of Bologna, Cesena Campus, Via dell'Universit\`a 50, Cesena, Italy}

%\begin{document}
%\maketitle

% \history{Date of publication xxxx 00, 0000, date of current version xxxx 00, 0000.}
% \doi{10.1109/ACCESS.2023.0322000}

% \title{Cellular Automata on Linear Groups
% }\author{
% \uppercase{Niccol\`o Castronuovo}\authorrefmark{1},
% \uppercase{Alberto Dennunzio}\authorrefmark{2}, 
% \uppercase{Luciano Margara}\authorrefmark{3},
% }

% \address[1]{Liceo ``A. Einstein'', Rimini, 47923, Italy (e-mail: castronuovoniccolo@gmail.com)}
% \address[2]{Dipartimento di Informatica, Sistemistica e Comunicazione, Universit\`a 
%   degli Studi di Milano-Bicocca,
%   Viale Sarca 336/14, 20126 Milano, Italy (e-mail: alberto.dennunzio@unimib.it)}

% \address[3]{Department of Computer Science and Engineering, University of Bologna, Cesena Campus, Via dell'Universit\`a 50, Cesena, Italy (e-mail: luciano.margara@unibo.it)}
% \tfootnote{``This work was partially supported by}

% \markboth
% {Dennunzio \headeretal: Cellular Automata on Matrix Groups}
% {Dennunzio \headeretal: Cellular Automata on Matrix Groups}

% \corresp{Corresponding author: Alberto Dennunzio (e-mail: alberto.dennunzio@unimib.it).}

 \begin{abstract}
%In recent years, significant research activity has focused on cellular automata defined over finite abelian groups  with the goal of solving one of the most challenging problems in cellular automata theory: explicitly establishing the relationship between the local rule defining the cellular automaton and its global dynamical behavior. Many fundamental properties have been fully characterized for this class of cellular automata in terms of easy to check properties of the local rule. The transition from abelian to non-abelian groups has introduced significant and unexpected technical complexities. Preliminary results have been recently obtained only for specific non-abelian groups, including simple, symmetric, alternating, dihedral, quaternion, and decomposable groups.

%\todo{riguardare}
We prove that many dynamical properties of group cellular automata  (\ie, cellular automata defined on any finite group and with global rule which is an endomorphism)—including surjectivity, injectivity, sensitivity to  initial conditions, strong transitivity,  positive expansivity, and topological entropy—can be decided by decomposing them into a set of much simpler group cellular automata. To be more specific, we provide a novel algorithmic  technique allowing one to decompose the group cellular automaton to be studied into a finite number of group cellular automata, some of them defined on abelian groups, while  others, if any, defined on products of simple non-abelian isomorphic groups. 
%It is worth noting that these groups are completely independent of the automaton itself, and therefore do not inherit any aspect of the complexity of its dynamics.
It is worth noting that the  groups resulting from the decomposition only depend  on the original group and therefore they are completely independent of both the automaton and the property under investigation. As a result, they do not inherit any aspect of the complexity of the automaton under investigation.

We prove that the group cellular automata obtained by the decomposition preserve dynamical properties and turn out to be much easier to analyze if compared to the original cellular automaton. As a consequence of these results, we show that injectivity, surjectivity and sensitivity to initial conditions are decidable properties and that no strongly transitive, and therefore no positively expansive, group cellular automata defined
on non-abelian groups exist. Moreover, we prove that the topological entropy of a group cellular automaton can be computed, provided we know how to compute the topological entropy for group cellular automata defined on products of simple non-abelian isomorphic groups and on abelian groups.
%In other words, our work demonstrates that the dynamics of a group cellular automaton can be analyzed by decomposing it into an ``abelian component" and a ``non-abelian component" on products of isomorphic simple groups.
\end{abstract}

%\begin{graphicalabstract}
%
%\end{graphicalabstract}

%%Research highlights
\begin{comment}
\begin{highlights}
\item We prove that many dynamical properties of group cellular automata (i.e., cellular automata defined on any finite group and with global rule which is an endomorphism) can be decided by decomposing them into a set of much simpler group cellular automata.
\item We provide a novel algorithmic technique allowing one to decompose the group cellular automaton to be studied into a finite number of group cellular automata, some of them defined on abelian groups, while others, if any, defined on products of simple non-abelian isomorphic groups.
\item We prove that the group cellular automata obtained by the decomposition preserve dynamical properties and turn out to be much easier to analyze if compared to the original cellular automaton.
\end{highlights}
\end{comment}
 % \begin{keywords}
 % Cellular Automata, Group Cellular Automata, Dynamical Behavior, Linear Groups, Matrix Groups. \end{keywords}

% \titlepgskip=-21pt

%\noindent {\bf Keywords:  Cellular Automata, Group Cellular Automata, Dynamical Behavior, Chaos, Decidability}  

%\noindent {\bf MSC2020:} ?? (primary); ?? (secondary).
\end{frontmatter}
%\end{document}

\section{Introduction}\label{secintro}

% \todo{ho definito il comando \textbackslash ZZ per denotare gli interi}
Cellular Automata (CA) serve as formal models for complex systems and can be viewed as discrete-time dynamical systems consisting of a regular grid of variables, each taking values from a finite set. The overall state of a CA (with a little abuse of notation, we will use CA also to denote a single cellular automaton) is defined by the values of all variables at a specific time $t$, and it evolves in discrete time steps according to a specified local rule. This rule updates in a synchronous and homogeneous way each variable on the basis of  the values at time $t-1$ of its neighboring variables (for an introduction to CA theory, see~\cite{hedlund69, Hadeler2017CellularAA, ceccherinigroupca}).

CA have been the subject of significant research and are applied in various fields, including computer science, physics, mathematics, biology, and chemistry, for purposes such as simulating natural phenomena, generating pseudo-random numbers, processing images, analyzing universal computation models, and cryptography (see for example~\cite{MARTINDELREY20051356,AnghelescuIS07,RubioEWRS04,kari2000,DennunzioFGM21INS}). %This is essentially due to three reasons: the huge variety of distinct CA dynamical behaviors; the emergence of complex behaviors from simple local interactions; the ease of their implementation (even at a hardware level).

One of the central challenges in CA theory is describing the global behavior of a CA based on the analysis of its local rule. While the local rule has a finite representation (e.g., a finite table), the global behavior can encompass an arbitrarily large, potentially infinite, amount of information. In fact, the grid of variables representing the global state of the CA may have any size and the desired global behavior might only emerge after an arbitrarily large number of time steps.

Many properties related to the temporal evolution of general CA have been proved to be undecidable, including non-trivial properties of their limit sets and fundamental dynamical properties like sensitivity to initial conditions, equicontinuity, topological transitivity, and chaos (see for example~\cite{ckps89,cks88,Hurd_Kari_Culik_1992,Kari94}). Since in practical applications one needs to know if the CA used for modelling a certain system exhibits some specific property, this can be a severe issue. 

Fortunately, the undecidability issue of dynamical properties of CA can be tackled by placing specific constraints on the model. In many cases, like the one we are exploring in this paper, the alphabet and the local rule are restricted to being a finite group and a homomorphism, respectively, giving rise to Group CA (GCA). It is important to note that these constraints do not at all hinder the effectiveness of such CA in practical applications. In fact, GCA can exhibit much of the complex behaviors of general CA and they are often used in various applications (see for example \cite{DennunzioFGM21INS,NandiKC94, RubioEWRS04}). 

During the last few decades, substantial efforts have been made to analyze the dynamical behavior of GCA on abelian finite groups, and more recently, on non-abelian (general) finite groups as well. %From now on, we will refer to CA on finite groups as GCA (again, both for the singular and plural forms).

Many fundamental global properties  of abelian GCA such as  injectivity, surjectivity, sentitivity to the  initial conditions, topological transitivity, ergodicity, positive expansivity, denseness of periodic orbits, and chaos have been fully characterized in terms of easy to check properties of their local rules (see \cite{CattaneoFMM00,DennunzioLFM09,ManziniM99a,ManziniM99} for GCA on $\ZZ/m\ZZ$ and \cite{DennunzioFGM2020TCS,Dennunzio20JCSS,DennunzioFMMP19,DBLP:journals/access/DennunzioFM23,DBLP:journals/isci/DennunzioFM24,kari2000} for GCA on general abelian groups). For GCA on $\ZZ/m\ZZ$, closed formulas for topological entropy and Lyapunov exponents were also  provided in \cite{DamicoMM03} and \cite{FinelliMM98}, respectively.
 
As to non-abelian GCA, preliminary results on injectivity and surjectivity were provided in \cite{BeaurK24} where the authors prove that, in dimensions $D>1$, injectivity and surjectivity remain decidable. Later, the dynamical behavior of GCA was investigated across various classes of finite groups, including simple, symmetric, alternating, dihedral, quaternion, and decomposable groups, with initial findings reported in \cite{GCA24}.

Surprisingly, the non-abelian nature of the group imposes significant restrictions on defining the local rule of GCA, resulting in a highly constrained class of GCA that is much more challenging to study. 

In this paper, we introduce a novel and general algorithmic technique that simplifies the analysis of the dynamical properties of \emph{any} GCA and the problem of deciding them, by reducing the GCA  to the study of a finite sequence of GCA, some of them defined on abelian groups and others on products of isomorphic simple groups. 

Specifically, we show  that several important properties (surjectivity, injectivity, sensitivity to  initial conditions, equicontinuity) hold for a GCA $\glorule$ defined on a finite group $G$ (not necessarily abelian) if and only if the same properties hold for a corresponding set $\{\glorule_1, \dots, \glorule_k\}$ of GCA, which are derived from $\glorule$ by means of our algorithmic decomposition techìnique and defined on much simpler finite groups $\{G_1, \dots, G_k\}$, where each  $G_i$ is either abelian or 
the product of simple non-abelian isomorphic groups. Using the same technique, we also show that the topological entropy of a GCA can be computed, provided we know how to compute the topological entropy for GCA defined on products of simple isomorphic groups and on abelian groups, which are two open but much simpler problems.

%(see Thms. \ref{surj_inj},\ref{transitive},\ref{strong_trans},\ref{pos_exp}).

As previously recalled, the dynamical properties of GCA on abelian groups have been fully characterized in terms of easy to check properties of their local rules, while, regarding GCA defined on the product of simple non-abelian isomorphic  groups, we show in this paper that their dynamical behaviors are highly constrained and  can be thoroughly analyzed, especially in the surjective case.
Hence, the study of the dynamics of the GCA $\glorule$ on a group $G$ can be effectively reduced to that of the GCA $\glorule_1$, \ldots, $\glorule_k$ some of them defined on abelian groups and others on products of isomorphic simple groups, thus significantly simplifying the problem of studying the dynamical behavior of $\glorule$. 

A crucial point is then how to get the set of $\{\glorule_1, \dots, \glorule_k\}$ of GCA defined on the finite groups $\{G_1, \dots, G_k\}$ starting from any given GCA on $G$. We show that the those sets can be obtained  by starting with the group $G$ and repeatedly applying the quotient operation by any non-trivial normal and fully invariant subgroup (invariant under the action of any group endomorphism), our algorithmic decomposition technique being just the implementation of all this and based on a divide and conquer strategy (see Algorithm 1 and 2 in Section~\ref{section_spezzatino}).  At each quotienting step, we obtain two smaller groups (and correspondingly two GCA): the fully invariant normal subgroup and the quotient group. The quotienting process ends when we either obtain an abelian group or reach groups that do not admit any non-trivial normal and fully invariant subgroups. According to the theory, we know that these groups are, in fact, products of isomorphic copies of a simple group.

Our decomposition technique allowed us to prove that injectivity, surjectivity and sensitivity to  initial conditions/equicontinuity are decidable properties for GCA and that there are no strongly transitive nor positively expansive GCA on non-abelian groups. 

% We  conjecture that no strongly transitive or positively expansive GCA exist on solvable non-abelian groups either.

We are confident that the same decomposition technique used for one-dimensional GCA could also be successfully applied in dimensions $D>1$.

\noindent The rest of this paper is organized as follows.

Section~\ref{pre_group} contains  the basic group theory results needed throughout the paper. In Section~\ref{cellular_a} we recall the fundamental definitions and known facts about CA and their topological and dynamical properties. Section~\ref{general_res} is devoted to the study of fundamental properties of GCA, such as injectivity, surjectivity, sensitivity to initial conditions, equicontinuity, topological transitivity, and topological entropy. We prove a series of results that allow us to reduce determining whether a given property holds for a GCA on a group $G$ to checking if the same property holds for two GCA both defined  on two smaller groups: a suitable normal subgroup  $H$ of $G$ and the quotient group $G/H$.
We also prove that a GCA on a non-abelian group cannot be neither strongly transitive nor positively expansive. Furthermore, we provide several results concerning the equivalence of dynamical, topological, and metric properties, which are typically distinct in the context of general CA. In Section~\ref{section_spezzatino}, we provide our decomposition technique  allowing one to exploit the results from Section~\ref{general_res} to reduce the study of a given property for a GCA to the analysis of a number of GCA defined on just two types of groups: abelian groups and products of isomorphic copies of simple non-abelian groups. We also provide the decidabilty results of the dynamical properties for GCA. 
%, thus significantly simplifying the problem. 
%we show how to apply the results from Section~\ref{general_res} to reduce the study of a given property for a GCA to the analysis of a number of GCA defined on just two types of groups: abelian groups and products of isomorphic copies of simple non-abelian groups, thus significantly simplifying the problem. 
Section~\ref{conclusion} contains some concluding remarks and a list of open questions. Finally,~\ref{A1} contains basic definitions and results on group and field theory.

\section{Preliminary results  on groups}\label{pre_group}

In this section, we state some known results and prove a few preliminary lemmata about groups that will be useful in the subsequent discussion.
Readers unfamiliar with group theory can refer to~\ref{A1}.

The following lemma is a trivial consequence of the definitions of commutator subgroup and center of a group. 

%[see \cite[p. 120]{kurosh1955theory}]
\begin{lemma}\label{Lemma_commutator_prod}
Let $G$ be a direct product of groups, \ie, $G=G_1\times\dots \times G_k$ for some integer $k>1$. It holds that the commutator subgroup of $G$ is the direct product of the commutator subgroups of $G_1$, $G_2,\dots, G_k$ and the center of $G$ is the direct product of the centers of $G_1$, $G_2,\ldots, G_k$.
\end{lemma}
A subgroup $H$ of a group $G$ is \textit{characteristic} if $\varphi(H)\leq H$ for every $\varphi\in \Aut(G)$.
As an example, the center of a group is a characteristic subgroup. 
A stronger property is  fully invariance. A \textit{fully invariant} subgroup of a group $G$ is a subgroup $H$ of $G$ such that, for every $\phi \in End(G)$, $\phi(H)\leq H$.

In other terms, the restriction of an endomorphism of $G$ to a fully invariant subgroup $H$ induces an endomorphism on $H$. We also recall that a fully invariant subgroup is normal.

\begin{lemma}\label{Lemma_commutator_fullyinv}
The commutator subgroup of any group is a fully invariant subgroup. More generally, if $G_1$ and $G_2$ are groups and $f\in \Hom(G_1,G_2)$, then the restriction of $f$ to the commutator of $G_1$ is an homomorphism from $[G_1,G_1]$ to $[G_2,G_2]$.
\end{lemma}
% \begin{proof}
% See \cite[p. 108]{Rotman_groups}
% \end{proof}

\begin{lemma}\label{pre_image_normal}
Let $G$ and $G'$ be two groups and let $f\in \Hom(G,G')$.
If $N'\trianglelefteq G'$ then the preimage  
$f^{-1}[N']$ is a normal subgroup of $G$.
If $f$ is surjective and $N\trianglelefteq G$, $f(N)$ is a normal subgroup of $G'$.
\end{lemma}
% \begin{proof}
% The proof is an almost trivial application of the definitions of normal subgroup and homomorphism. 
% \end{proof}

% \todo{se vogliamo dimostrare in dettaglio notare nella prima parte che $\phi^{-1}(N')$ \'e il ker di  $\pi\circ f$ dove $\pi$ proietta sul quoziente per $N'$.}

Now we introduce a notation that will be systematically used in the following. Let $e$ be the identity element of a group $G$. 
If $f\in \Hom(G^k,G)$, we will write $f=(h_1,\ldots,h_k)$ in which $h_i\in \End(G)$ is defined, for all $g\in G$, as
$h_i(g)=f(e,\ldots,e,g,e,\ldots,e)$, where $g$ occupies the $i$-th position in $f$ and all the other entries are equal to $e$.

\begin{lemma}\label{surj_normal}
Let $\locrule\in \Hom(G^k,G)$ with $\locrule=(h_1,\ldots,h_k)$.
Let $N\trianglelefteq G$. 
Then, for every $i$, $h_i(N)$ is normal in $\Imma(f)$.
In particular, if  $\locrule$ is surjective, every $\Imma(h_i)$ is normal in $G$.
\end{lemma}
\begin{proof}
This follows by Lemma \ref{pre_image_normal} and the fact that if $N\trianglelefteq G$ then the subgroup $M_i$ of $G^k$ consisting of the $k-$tuples with the $i$-th component in $N$ and all the others equal to $e$ is a normal subgroup of $G^k$ and $h_i(N)=f(M_i)$.
\end{proof}

% Consider $h_1(a_1)h_2(a_2)\ldots h_k(a_k)\in \Imma(\locrule)$ and $h_i(b)\in h_i(N)$. We have to show that $$h_1(a_1)h_2(a_2)\ldots h_k(a_k)h_i(b)(h_1(a_1)h_2(a_2)\ldots h_k(a_k))^{-1}\in h_i(N).$$

% Since $\Imma(h_i)\subseteq C_{\GL}(\Imma(h_j))$ for every $i\neq j$, this reduces to $h_i(a_i)h_i(b)h_i(a_i^{-1})=h_i(a_iba_i^{-1})\in h_i(N)$. Since $N$ is normal in $G$, $a_iba_i^{-1}\in N$ and $h_i(a_iba_i^{-1})\in h_i(N)$.

\begin{lemma}\label{f_sur_preserva_Z}
Let $\locrule\in \Hom(G^k,G)$ be a surjective homomorphism. 
Then, $f(\Z^k_G)\subseteq \Z_G$.
In particular, if $f=(h_1,\ldots,h_k)$, it holds that $h_j(\Z_G)\subseteq \Z_G$ for every $j$.
\end{lemma}
\begin{proof}
Consider $(a_1,\ldots,a_k)\in \Z^k_G$ and $(b_1,\ldots,b_k)\in G^k$. We get 
\begin{align*}
f(a_1,\ldots,a_k)f(b_1,\ldots,b_k)&=f(a_1b_1,\ldots,a_kb_k)\\&=f(b_1,\ldots,b_k)f(a_1,\ldots,a_k)\enspace.
\end{align*}
Since $f$ is surjective this implies that $f(a_1,\ldots,a_k)\in \Z_G$.
In particular, $f(e,\ldots,e,a_i,e,\ldots,e)=h_i(a_i)\in \Z_G$ for every $a_i \in \Z_G.$
\end{proof}

% \todo{ma il lemma qui sopra serve ? il centro è fully invariant di suo, perché assumiamo $f$ suriettiva? }

\section{About CA and GCA} \label{cellular_a}

%\todo{C'erano due notazioni per la posizione i esima di una configurazione: $c(i)$ e $c_i$. Ho messo dovunque $c_i.$}
In this section, we review the fundamental definitions and key results related to CA and GCA. For additional definitions and results, we refer the reader to those introduced in~\cite{GCA24}.
\subsection{CA configurations} 
Let $G$ be a finite set. A CA \emph{configuration} is any function from $\ZZ$ to $G$, \ie, an element of $\gzd$.  Given a configuration $c\in\gzd$ and any integer $i\in\ZZ$, the value of $c$ at position $i$ is denoted by $c_i$, while for any $i,j\in\ZZ$ with $i\leq j$ we note $c_{[i,j]}=c_i\ldots c_j\in G^{j-i+1}$.

\begin{comment}
A \emph{CA} on $G$ is any continuous function $\glorule: G^\ZZ \to G^\ZZ$ which also shift commuting, \ie, $\glorule\circ \sigma=\sigma \circ \glorule$, where the shift map $\sigma: G^\ZZ\to G^\ZZ$ is defined as follows
 $$\forall c\in G^\ZZ, \forall i\in \ZZ:\  \sigma(c)_i=c_{i-1}.
 $$
 %and the continuity of $\glorule$ refers to the product topology induced by the discrete topology on $G$.
%  Any CA  can be equivalently defined by means of a \emph{local rule} $f:G^k \to G$ (see~\cite{hedlund69}) that is paired with an ordered integer vector $v\in \ZZ^k$ called the \emph{neighbor vector}. Namely, a CA $\glorule$ based on $(f,v)$ is defined as follows:
% $$\forall c \in G^\ZZ,\ \forall i\in \ZZ:\ \glorule(c)(i)=f(c(i+v_1),\dots, c(i+v_k))\enspace$$
% The \textit{radius} of the CA $\glorule$ is defined as $$\rho(\glorule)=\max_{j\in v}\{|j|\}.$$
Any CA  can be equivalently defined by means of a \emph{local rule} $f:G^{2\rho +1} \to G$, where $\rho\in \N$ (see\cite{hedlund69}). Namely, a CA $\glorule$ with local rule  $f$ is defined as follows:
$$\forall c \in G^\ZZ,\ \forall i\in \ZZ:\ \glorule(c)_i=f(c_{i-\rho} ,\dots, c_{i+\rho}).$$ 
The set $G^{\mathbb Z}$ is equipped with the standard Tychonoff ultrametric $d$ 
defined as 
$$
\forall c,c^\prime\in\gzd:\  d(c,c')=\begin{cases}
0 & \text{ if }  c=c' \\
2^{-\Delta(c,c')} &\text{otherwise}
\end{cases}%,\;\quad \text{where}\;n=\min\{|\v|\,:\,c(\v)\nec'(\v) \}\enspace.
$$
where $\Delta(c,c')=\min\{|j|: j\in\mathbb Z \text{ and } c_j\neq c'_j\}$.
\end{comment}
The set $G^{\mathbb Z}$ is also a topological space with the \textit{prodiscrete topology}, i.e., the product topology when each factor $G$ is given the discrete topology. For any $i,j\in\ZZ$ with $i\leq j$ and any $u\in G^{j-i+1}$, the \textit{cylinder} $C([i,j],u)$ 
%where $[i,j]$ is the interval $i,i+1,\ldots ,j$ and $\pi=\pi_i\pi_{i+1}\cdots \pi_{j}\in G^{j-i+1},$ 
is the subset of $G^{\mathbb Z}$ defined as $$C([i,j],u):= \{c\in\gzd: c_{[i,j]}=u\}$$
The cylinders form a clopen basis for the prodiscrete topology and, when equipped with that topology, $\gzd$ turns out to be a  a compact, Hausdorff, and totally disconnected topological space. Moreover,  $\gzd$ is Polish space, i.e., a separable completely metrizable topological space. 
%\bigcap_{s\in [i,j]} C(s,\pi_s).$$ 
\begin{comment}
An \textit{elementary cylinder} is a subset of $G^{\mathbb Z}$ the form  $$C(j,a)=\{c\in G^{\mathbb Z}:\  c_j=a\}$$ with $j\in \mathbb Z$ and $a\in G$.
A \textit{cylinder} is any finite intersection of elementary cylinders. The cylinder $C([i,j],\pi),$ where $[i,j]$ is the interval $i,i+1,\ldots ,j$ and $\pi=\pi_i\pi_{i+1}\cdots \pi_{j}\in G^{j-i+1},$ is defined as $$C([i,j],\pi):=\bigcap_{s\in [i,j]} C(s,\pi_s).$$ 
\end{comment}
Indeed, the set $G^{\mathbb Z}$ can be equipped with the standard Tychonoff ultrametric $d$ 
defined as 
$$
\forall c,c^\prime\in\gzd:\  d(c,c')=\begin{cases}
0 & \text{ if }  c=c' \\
2^{-\Delta(c,c')} &\text{otherwise}
\end{cases}%,\;\quad \text{where}\;n=\min\{|\v|\,:\,c(\v)\nec'(\v) \}\enspace.
$$
where $\Delta(c,c')=\min\{|j|: j\in\mathbb Z \text{ and } c_j\neq c'_j\}$ and the topology induced by the Tychonoff metric coincides with the prodiscrete topology. Since it has no isolated points, the set $\gzd$ is also a Cantor space. 

\subsection{CA} A \emph{CA} on $G$ is any continuous function $\glorule: G^\ZZ \to G^\ZZ$ which also shift commuting, \ie, $\glorule\circ \sigma=\sigma \circ \glorule$, where the shift map $\sigma: G^\ZZ\to G^\ZZ$ is defined as follows
 $$\forall c\in G^\ZZ, \forall i\in \ZZ:\  \sigma(c)_i=c_{i-1}.
 $$
 %and the continuity of $\glorule$ refers to the product topology induced by the discrete topology on $G$.
%  Any CA  can be equivalently defined by means of a \emph{local rule} $f:G^k \to G$ (see~\cite{hedlund69}) that is paired with an ordered integer vector $v\in \ZZ^k$ called the \emph{neighbor vector}. Namely, a CA $\glorule$ based on $(f,v)$ is defined as follows:
% $$\forall c \in G^\ZZ,\ \forall i\in \ZZ:\ \glorule(c)(i)=f(c(i+v_1),\dots, c(i+v_k))\enspace$$
% The \textit{radius} of the CA $\glorule$ is defined as $$\rho(\glorule)=\max_{j\in v}\{|j|\}.$$
Any CA  can be equivalently defined by means of a \emph{local rule} $f:G^{2\rho +1} \to G$, where $\rho\in \N$ (see\cite{hedlund69}). Namely, a CA $\glorule$ with local rule  $f$ is defined as follows:
$$\forall c \in G^\ZZ,\ \forall i\in \ZZ:\ \glorule(c)_i=f(c_{i-\rho} ,\dots, c_{i+\rho}).$$ 

A CA $\glorule$ is said to be \textit{injective} (\textit{surjective})  if the map $\glorule$ is injective (surjective). We recall that injective CA are surjective and a CA is surjective if and only if every configuration has a finite and uniformly bounded number of pre-images~\cite{hedlund69}. 
A CA $\glorule$ is said to be \textit{open} if $\glorule$ is an open map with respect to the prodiscrete topology, i.e., if it maps open sets to open sets. 
\smallskip

A CA  $\glorule$ is \emph{topologically transitive} if for any pair of nonempty open subsets $U,V\subseteq\gzd$ there exists a natural $t>0$ such that $\glorule^t(U)\cap V\neq\emptyset$, while $\glorule$ is \emph{strongly transitive} if for any nonempty open subset $U\subseteq\gzd$ it holds that $\bigcup_{t\in\N} \glorule^t(U)=X$. Strongly transitive CA are topologically transitive and topologically transitive CA are surjective. Strongly transitive CA cannot be injective. 
%A CA $(\gzd,\glorule)$ has \emph{dense periodic orbits} if the set of its periodic points is dense in $\gzd$, where a periodic point for is any configuration $c\in\gzd$ such that $F^h(c)=c$ for some natural $h>0$.

%A map $\glorule$ over a metric space $M$ is said to be \textit{sensitive to initial conditions} (\textit{sensitive} for short) if 
%$$\exists \epsilon >0\;\mbox{s.t.}\;\forall x \in M,\;\forall \delta>0\; \exists y\in M \mbox{ with }d(x,y)<\delta\mbox{ and }\exists n\geq 0\mbox{ s.t. }d(\glorule^n(x),\glorule^n(y))\geq \epsilon. $$
 
A CA  $\glorule$ is \emph{sensitive to  initial conditions}  if there exists $\epsilon>0$ such  that for any $\delta>0$ and $c\in G^{\mathbb Z}$ there is a configuration $c'\in G^{\mathbb Z}$ with $0<d(c',c)<\delta$ such that $d(\glorule^t(c'), \glorule^t(c))\geq\epsilon$ for some natural $t$.
Sensitivity to  initial conditions  is the well-known basic component and essence of the chaotic behavior of discrete time dynamical systems.

%Let $M$ be any topological space and let $F:M\to M$  be any function. An element $x \in M$ is a \textit{periodic point} if $F^h(x) = x$ for some natural $h > 0$. The dynamical system $(M ,F)$ has \textit{dense periodic orbits} (DPO) if the set of its periodic points is dense in $M$.  

%A CA  $\glorule$ has \emph{dense periodic orbits} if the set $$PO=\{c\in \gzd:\ \glorule^n(c)=c \text{ for some natural } n>0 \}$$ is dense in $\gzd$.

A CA $\glorule$ has \emph{dense periodic orbits} (DPO) if the set of its periodic points is dense in $\gzd$, where a periodic point for $\glorule$ is any configuration $c\in\gzd$ such that $F^h(c)=c$ for some natural $h>0$.

Sensitivity to  initial conditions, topological transitivity and DPO are the features that together define the popular notion of \emph{chaos} according to the Devaney definition \cite{devaney89}.

A CA  $\glorule$ is said to be  \emph{equicontinuous} if for any $\epsilon>0$ there exists $\delta>0$ such that for all $c,c'\in\gzd$, $d(c,c')<\delta$ implies that $\forall k\in\N$, $d(\glorule^k(c'), \glorule^k(c'))<\epsilon$.

Note that there are CA that are neither sensitive to  initial conditions nor equicontinuous, and these are called almost equicontinuous. In the case of CA on abelian groups, however, this intermediate case does not exist: CA on abelian groups are either equicontinuous or sensitive to initial conditions.

%A map $\glorule$ over a metric space $M$ is said to be \textit{equicontinuous at a point} $x\in M$ if 

%$$\forall \epsilon >0\;\exists \delta>0\;\mbox{s.t.}\; \forall  y\in M\mbox{ with }d(x,y)<\delta\mbox{ and }\forall n\geq 0,\mbox{ }d(\glorule^n(x),\glorule^n(y))< \epsilon. $$

%Clearly, $\glorule$ is not sensitive if and only if there exists an $x\in M$ such that $\glorule$ is equicontinuous at $x$.

%A map $\glorule$ over a metric space $M$ is said to be \textit{pointwise equicontinuous} if it is equicontinuous at every point. 

%A map $\glorule$ over a metric space $M$ is said to be \textit{uniformly equicontinuous} if $$\forall \epsilon >0\;\exists \delta>0\;\mbox{s.t.}\; \forall  x,y\in M \mbox{ with }d(x,y)<\delta\mbox{ and }\forall n\geq 0,\mbox{ }d(\glorule^n(x),\glorule^n(y))< \epsilon. $$

%Notice that pointwise equicontinuity and uniform equicontinuity are not equivalent since the $\delta$ in the definition of equicontinuity at $x$ depends on $x$. However the two definitions coincides over compact spaces (see \cite{Elaydi_Farran_1990}), thus we will simply speak of \textit{equicontinuous} CA. 
A CA $\glorule$ is \emph{positively expansive} if for some constant $\varepsilon>0$ it holds that for any pair of distinct configurations $c,c'\in G^{\mathbb Z}$ there exists a natural number $t$ such that $d(\glorule^{t}(c),\glorule^{t}(c'))\geq\varepsilon$. We emphasize that positive expansivity is a strong form of chaos.  Indeed, on one hand, positive expansivity for a CA is a stronger condition than sensitivity to initial conditions. On the other hand, any positively expansive CA is also strongly transitive (and thus topologically transitive), and at the same time, it has DPO. Therefore, any positively expansive CA is chaotic according to Devaney's definition of chaos. Clearly, if a CA 
$\glorule$ is positively expansive, then it is surjective but not injective.

%\todo{serve la def di permutative ?}
%A map $f:G^k \to G$ is said to \emph{permutative in the variable $i\in\{1, \ldots, k\}$} iff for every $(a_1, \ldots, a_{i-1}, a_{i+1}, \ldots, a_k)\in G^{k-1}$ and every $b\in G$ there exists a unique $a\in G$ such that $f(a_1, \ldots, a_{i-1}, a, a_{i+1}, \ldots, a_k)=b$. CA defined by permutative local rule have been studied in depth. In particular, CA defined by a rightmost (resp., leftmost permutative) local rule and by a neighbor  vector $v$ such that $v_k>0$ (resp., $v_k<0$) turn out to be topologically transitive and then surjective. 
\smallskip

\textit{Topological entropy} is one of the
most studied properties of dynamical systems. Informally, the topological entropy measures the uncertainty of the forward evolution of any dynamical system in
the presence of incomplete description of initial configurations.
The definition of topological entropy $\entropy{H}(F)$ of a continuous map $F:X\to X$ over a compact space $X$ was introduced in \cite{Adler}.
For 1-dimensional CA $\glorule$, the  definition of topological entropy is equivalent to the following \cite{Hurd_Kari_Culik_1992}. Let $R(w,t)$ denote the number of distinct rectangles of width $w$ and height $t$ occurring in a space-time evolution diagram of $\glorule$ (for the exact definition of $R(w,t)$ see \cite{Hurd_Kari_Culik_1992}). Then, 
$$\entropy{H} (\glorule)=\lim_{w\to +\infty}\lim_{t\to +\infty}\dfrac{\log R(w,t)}{t}.$$

%\todo{io uso $\Z$ per il centro di un gruppo, Luciano per gli interi. ho reso coerenti le notazioni}

\subsection{GCA}
Let $G$ be a finite group with identity element $e$. The set $G^\ZZ$ is also a group, with the component-wise operation defined by the group operation of $G$, and we denote by $e^{\ZZ}$ the configuration taking the value $e$ at every integer position, \ie, $e^{\ZZ}$ is the identity element of the group $G^\ZZ$. Clearly, when equipped with the prodiscrete topology, $G^\ZZ$ turns out to be both a profinite and Polish group. A configuration $c\in\gzd$ is said to be \emph{finite} if the number of positions $i\in\ZZ$ such that $c_i\neq e$ is finite.

The (normalized) Haar measure $\mu$ on $G^{\mathbb Z}$ is the product measure induced by the uniform probability distribution
over $G$. In particular, for every  cylinder $C([i,j],u)$ we have $\mu(C([i,j],u))=\frac{1}{|G|^{j-i+1}}$.
If $H\leq G$, then $H^{\mathbb Z}$ is a closed subgroup of $G^{\ZZ}$. Moreover, $H\trianglelefteq G$ if and only if $H^{\mathbb Z}\trianglelefteq G^\ZZ$.
In this case, the prodiscrete topologies on $H^{\mathbb Z}$ and $(G/H)^{\mathbb Z}$ agree with the subspace topology on $H^{\mathbb Z}$ and with the quotient topology on $(G/H)^{\mathbb Z}$, respectively.

A CA $\glorule: \gzd\to\gzd$ is said to be a \emph{GCA} if $\glorule$ is an endomorphism of $G^{\ZZ}$. In that case, the local rule of $\glorule$ is a homomorphism $f:G^{2\rho+1} \to G$ (see~\cite{CR22} for a proof as far as an arbitrary algebraic structure is concerned).
The \emph{kernel of a GCA} $\glorule$ is $Ker(\glorule)=\{c\in\gzd: \glorule(c)=e^{\ZZ}\}$.   

Given any function $f:G^{2\rho+1}\to G$, by \cite[Thm.1]{GCA24}, $f\in \Hom(G^{2\rho+1},G)$ if and only if there exist $2\rho+1$ endomorphisms $h_{-\rho},\ldots,h_{\rho}\in \End(G)$, such that $f(g_{-\rho},\ldots,g_{\rho})=h_{-\rho}(g_{-\rho}) \cdots h_{\rho}(g_{\rho})$ for all $g_{-\rho},\ldots,g_{\rho}\in G$ and 
$\Imma(h_i)\subseteq C_G(\Imma(h_j))$ for every $i\neq j$.
In this case, according to the notation introduced in Section~\ref{pre_group}, we will write $f=(h_{-\rho},\ldots,h_{\rho})$. 
Notice that some of the $h_i$'s could be trivial but we will always assume, if not otherwise stated, that at least one between $h_{-\rho}$ and $h_{\rho}$ is non-trivial. In this case ${\rho}$ will be said to be \textit{the radius} of the GCA $\glorule$ and will be indicated also with $\rho(\glorule).$  

A GCA with local rule $f$  is called \textit{shift-like} if $\rho\geq 1$ and only one between the $h_i$'s is non trivial, while  it is called \textit{identity-like} if $\rho=0$. Surjective shift-like GCA are topologically transitive. 

We recall here the following Theorem (which is a slight generalization of \cite[Thm.5]{GCA24}) that states that, for simple non-abelian and quasi-simple groups, the structure of GCA is almost trivial. 

\begin{theorem}
Let $G$ be a simple  non-abelian or a quasi-simple  group. Then, any GCA on $G$ is either shift-like or identity-like. 
\end{theorem}
\begin{proof}
It follows from the fact that every endomorphism of a simple or a quasi-simple group is trivial or an automorphism. 
\end{proof}
%Now we summarizes the topological and metric properties of the group $G^{\mathbb Z}$.
%
% For the basic concepts of topology, topological group, and measure theory, see \cite{bredon1993topology} and \cite{federer2014geometric}. 
%
%\begin{lemma}
\begin{comment}
If $G$ is a finite group,
 the group $G^\ZZ$ with the prodiscrete topology is a profinite group, i.e., a compact, Hausdorff and totally disconnected topological group, and a Polish group, i.e., a separable completely metrizable topological group. 
The prodiscrete topology coincides with the topology induced by the Tychonoff metric and $G^{\mathbb Z}$ has no isolated points. 
The cylinders form a clopen basis for this topology. 
The (normalized) Haar measure $\mu$ on $G^{\mathbb Z}$ is the product measure induced by the uniform probability distribution
over $G$. In particular, for every elementary cylinder $C$, $\mu(C)=1/|G|$.
If $H\leq G$, then $H^{\mathbb Z}$ is a closed subgroup of $G^{\ZZ}$. Moreover, $H\trianglelefteq G$ if and only if $H^{\mathbb Z}\trianglelefteq G^\ZZ$.
In this case, the prodiscrete topologies on $H^{\mathbb Z}$ and $(G/H)^{\mathbb Z}$ agree with the subspace topology on $H^{\mathbb Z}$ and with the quotient topology on $(G/H)^{\mathbb Z}$, respectively.  
Any GCA $\glorule$ over $G$ is a continuous endomorphisms of the group $G^\ZZ$. 
\end{comment}
\section{Dynamical properties of GCA through group quotients}\label{general_res}

In this section, we show how to reduce the study of a number of dynamical properties of any GCA on a group $G$ to that of two GCA defined on two smaller groups: a normal subgroup $H$ of $G$ and the quotient group
$G/H$.
\smallskip

In the following, if $G$ is a group and $H\trianglelefteq G,$ the coset of the element $x\in G$ in $G/H$ will be  $[x]:=xH$. Moreover, if $c=(\ldots c_{-1}c_0c_1\ldots)\in G^{\ZZ}$, then $[c]$ will denote the element of $(G/H)^{\mathbb Z}$ given by $(\ldots [c_{-1}][c_0][c_1]\ldots)$. 

\begin{definition}\label{tilde_barra_def}
Let $G$ be a finite group and $H\trianglelefteq G$. Let $\glorule$ be a GCA on $G$ such that $\glorule(H^{\mathbb Z})\subseteq H^{\mathbb Z}$. Then, the maps 
\begin{eqnarray*}
&&\overline{\glorule}_{H}:H^{\mathbb Z}\to H^{\mathbb Z}  \text{ and}\\
&&\widetilde{\glorule}_{H}:(G/H)^{\mathbb Z}\to (G/H)^{\mathbb Z}
\end{eqnarray*}
%$\overline{\glorule}:H^{\mathbb Z}\to H^{\mathbb Z}$ and $\widetilde{\glorule}:(G/H)^{\mathbb Z}\to (G/H)^{\mathbb Z}$ 
are defined as follows:
%$(1)$ for every $c\in H^{\mathbb Z}:\ \overline{\glorule}(c):=\glorule(c)$, and \\
%$(2)$ for every $[c]\in G/H:\ \widetilde{\glorule}([c]):=[\glorule(c)]$.
\begin{eqnarray}
&&\forall c\in H^{\mathbb Z}:\ \overline{\glorule}_H(c):=\glorule(c) \text{ and }\label{Fbar}\\
&&\forall [c]\in (G/H)^{\mathbb Z}:\ \widetilde{\glorule}_H([c]):=[\glorule(c)]. \label{Ftilde}
\end{eqnarray}
\end{definition}

Note that, by Equations~\eqref{Fbar} and~\eqref{Ftilde} and since $\glorule(H^{\mathbb Z})\subseteq H^{\mathbb Z}$,  $\overline{\glorule}_H$ and $\widetilde{\glorule}_H$   are well-defined GCA on
$H^{\mathbb Z}$ and $(G/H)^{\mathbb Z}$, respectively. From now on, when the group $H$ is clear from the context, we will simplify the notation by using $\overline{\glorule}$ and $\widetilde{\glorule}$ in place of $\overline{\glorule}_{H}$ and $\widetilde{\glorule}_{H}$.
%$\overline{\glorule}$ and $\widetilde{\glorule}$ are defined based on a group $G$, a normal subgroup $H$ of $G$, and a GCA $\glorule$ on $G$ such that $\glorule(H^{\mathbb Z})\subseteq H^{\mathbb Z}$. When clear from the context, we will use $\overline{\glorule}$ and $\widetilde{\glorule}$ without explicitly defining $G$, $H$, and $\glorule$.
%\begin{proposition}\label{tilde_barra_prop}
%$\widetilde{\glorule}$ and $\overline{\glorule}$ are well defined GCA on 
%$H^{\mathbb Z}$ and $(G/H)^{\mathbb Z}$, respectively.
%\end{proposition}
%\begin{proof}
%The thesis follows from the fact that $\glorule(H^{\mathbb Z})\subseteq H^{\mathbb Z}$.
%\end{proof}
Clearly, it holds that if $\glorule$ is a GCA on the group $G$ with local rule $f=(h_{-\rho},\ldots,h_{\rho})$ and $H\trianglelefteq G$, then $\glorule(H^{\mathbb Z})\subseteq H^{\mathbb Z}$ if and only if $f(H^{2\rho+1})\subseteq H $ if and only if $h_i(H)\subseteq H$ for every $-\rho\leq i\leq \rho.$ 

%For later use, we point out that, 
We also stress that if  $f=(h_{-\rho},\ldots, h_{\rho})$ is the  local rule of $\glorule$, then the local rules $\overline{f}$ and $\widetilde{f}$ of $\overline{\glorule}$ and $\widetilde{\glorule}$ are $\overline{f}=(h_{-\rho}|_H,\ldots,h_{\rho}|_H)$ and $\widetilde{f}=(\widetilde h_{-\rho},\ldots,\widetilde h_{\rho})$ where $h_i|_H$ is the restriction of $h_i$ to $H$ and $\widetilde{h}_i([x])=[h_i(x)]$ for every $x\in G.$

\vspace{0.5cm}

Let us begin considering the two properties that are perhaps the most widely known and studied: surjectivity and injectivity.

%Now we prove a theorem that reduces the problem of determining if a GCA is injective or surjective to the same problem on two GCA over smaller groups.  

\begin{theorem}\label{surj_inj}
Let $G$ be a finite group and $H\trianglelefteq G$. Let $\glorule$ be a GCA on $G$ such that $\glorule(H^{\mathbb Z})\subseteq H^{\mathbb Z}$. Then, $\glorule$ is surjective (resp., injective) if and only if both $\widetilde{\glorule}$ and $\overline{\glorule}$ are surjective (resp., injective).
\end{theorem}

\begin{proof}
%To prove the first assertion of the Theorem it is sufficient to observe that $\widetilde{\glorule}$ is well defined since $\glorule(H^{\mathbb Z})\subseteq H^{\mathbb Z}$.
%Now we prove the second assertion. 

It is well known \cite{hedlund69} that a GCA $\glorule$ is surjective if and only if $\Ker(\glorule)$ is finite and is injective if and only if $\Ker(\glorule)$ contains only $e^{\mathbb Z}$. In particular, if $\glorule$ is injective it is also surjective. 
Clearly $\Ker(\overline{\glorule})\subseteq \Ker(\glorule)$. 
% \todo{il ker di Ftilde sono scatole mentre il ker di F sono elementi, forse dobbiamo sostituire $\Ker(\overline{\glorule})\subseteq \Ker(\glorule)$ con
% $|\Ker(\overline{\glorule})|\leq |\Ker(\glorule)|$ ?}
Thus, if $\glorule$ is surjective then $\overline{\glorule}$ is surjective and if $\glorule$ is injective then $\overline{\glorule}$ is injective. 
Moreover,
\begin{eqnarray*}
    \Ker(\widetilde{\glorule})&=&\{[c]\in \left(G/H\right)^\ZZ: \,\widetilde{\glorule}([c])=[e]^{\ZZ}\}\\
&=&\{[c]\in \left(G/H\right)^\ZZ: \, [\glorule(c)]=[e]^{\ZZ}\}\\
&=&\{[c]\in \left(G/H\right)^\ZZ: \,\exists h\in H^{\mathbb Z}, \,\glorule(c)h=e^{\ZZ}\}.
\end{eqnarray*}
Suppose now that $\glorule$ is surjective. 
Since in this case $\overline{\glorule}$ is also surjective, for every  $h\in H^{\mathbb Z}$ there exits $h'\in H^{\mathbb Z}$ such that $\glorule(h')=h$. Thus, the previous set can be written as
$$\{[c]\in \left(G/H\right)^\ZZ: \,\exists h'\in H^{\mathbb Z}, \glorule(ch')=e^{\ZZ}\}.$$
Hence, for every $[c]\in \Ker(\widetilde{\glorule})$ there exists an element $ch'\in \Ker(\glorule)$. Notice that if two elements $[c_1]$ and $[c_2]$ of $\Ker(\widetilde{\glorule})$ give rise to the same element $c_1h_1'=c_2h_2'$ of $\Ker(\glorule)$, then $c_1=c_2h_2'(h_1')^{-1}$ and so $[c_1]=[c_2]$. 
As a consequence, there exists an injective map from $\Ker(\widetilde{\glorule})$ to $\Ker(\glorule)$. 
Since $\glorule$ is surjective, $\Ker(\glorule)$ is finite. Thus, $\Ker(\widetilde{\glorule})$ is also finite and $\widetilde{\glorule}$ is surjective. 
If $\glorule$ is also injective, $\Ker(\glorule)=\{e^{\mathbb Z}\}$ and, hence, we get that $\Ker(\widetilde{\glorule})=\{[e]^{\mathbb Z}\}$.
Thus, $\widetilde{\glorule}$ is injective. 

Assume now that $\widetilde{\glorule}$ and $\overline{\glorule}$ are surjective. We want to show that $\glorule$ is surjective. 
Let $c\in G^{\mathbb Z}$. Since $\widetilde{\glorule}$ is surjective there exists $[d]\in \left(G/H\right)^{\mathbb Z}$ such that $\widetilde{\glorule}([d])=[c]$.
This is equivalent to state that there exits $h\in H^{\mathbb Z}$ with $\glorule(d)h=c$. Since $\overline{\glorule}$ is surjective, there exits $h'\in H^{\mathbb Z}$ such that $\glorule(h')=h$. In this way we get $\glorule(dh')=c$ and, therefore, $\glorule$ is surjective. 

To complete the proof we have to show that if $\widetilde{\glorule}$ and $\overline{\glorule}$ are injective then $\glorule$ is injective.
Let $c\in G^{\mathbb Z}$ such that $\glorule(c)=e^{\mathbb Z}$. Then, $\widetilde{\glorule}([c])=[e]^{\mathbb Z}$. Since $\widetilde{\glorule}$ is injective this implies that $[c]=[e]^{\mathbb Z}$, i.e., $c\in H^{\mathbb Z}$.
Thus $\overline{\glorule}(c)=\glorule(c)=e^{\mathbb Z}$ and, since $\overline{\glorule}$ is injective, $c=e^{\mathbb Z}$. This proves that $\glorule$ is injective. 
\end{proof}

% \todo{
% Proof alternativa che, se funziona,  forse funziona anche nel caso D-dimensionale:\\
% 1. $ker(F)$ sta tutto dentro a $H^{\ZZ^D}$  o sta tutto (tranne $e^{\ZZ^D}$) fuori da $H^{\ZZ^D}$. Infatti, siano $x$ e $y$ due elementi diversi da $e^{\ZZ^D}$ tali che $\glorule(x)=$\\
% 1. $ker(F)$ sta tutto dentro a $H^{\ZZ^D}$  o sta tutto (tranne $e^{\ZZ^D}$) fuori da $H^{\ZZ^D}$. Infatti, siano $x$ e $y$ due elementi diversi da $e^{\ZZ^D}$ tali che $\glorule(x)=\glorule(y)=e^{\ZZ^D}$\\ allora $xy$ appartiene al kernel di $\glorule$ 
% 2. se sta tutto fuori abbiamo $G^{\ZZ^D}\cong H^{\ZZ^D}\times G/H^{\ZZ^D}$ e quindi
% $\glorule\cong \overline{\glorule}\times \widetilde{\glorule}$\\
% 3. se sta tutto dentro allora  $\overline{\glorule}(H^{\ZZ^D})\subseteq H^{\ZZ^D}$ e
% $\widetilde{\glorule}(G/H^{\ZZ^D})\subseteq G/H^{\ZZ^D}$\\
% 4. Assumiamo vero il punto 3 e forse si riesce a dimostrare il teorema ...
% }

%Now we prove a lemma that will be useful in the following and which has a corollary interesting in itself. 

In the next lemma, we prove that if a GCA 
$\glorule$ defined on a finite group
$G$ is surjective, then the kernel of $\glorule$ is entirely contained in $\Z_G^\ZZ$.

\begin{lemma}\label{ker_sta__nel_centro_se_sur}
Let $\glorule$ be a GCA on a finite group $G.$ If $\glorule$ is surjective then $\Ker(\glorule)\subseteq \Z_G^\ZZ.$
\end{lemma}
\begin{proof}
Let $\glorule$ be a GCA on a finite group $G.$
The kernel of $\glorule$ is clearly $\sigma$-invariant, i.e., if $c\in \Ker(\glorule)$ then $\sigma(c)\in \Ker(\glorule).$
If $\glorule$ is surjective, then $\Ker(\glorule)$
is a finite set. 
As a consequence, every element $c\in \Ker(\glorule)$  must be $\sigma$-periodic, i.e., there exists $m\in \mathbb Z$ such that $\sigma^m(c)=c.$
Let $c\in \Ker(\glorule)$ and assume by contradiction that there exists $i \in \mathbb Z$ such that $c_i\not\in \Z_G.$
Then, there exists $g_i\in G$ such that $g_i^{-1}c_ig_i\neq c_i.$
Consider the configuration $g\in G^{\mathbb Z}$ with $g_i$ in position $i$ and $g_j=e$ for every $j\neq i$. Since $\Ker(\glorule)$ is a normal subgroup of $G^{\mathbb Z}$,
then $g^{-1}cg\in \Ker(\glorule)$. Since $g^{-1}cg$ is  not $\sigma$-periodic, we get a contradiction. 
\end{proof}

Notice that, if a GCA $\glorule$ is surjective, its local rule $\locrule$ is surjective and therefore $\glorule(\Z_G^{\mathbb Z})\subseteq \Z_{G}^\ZZ$ by Lemma \ref{f_sur_preserva_Z}. Thus, if $\glorule$ is surjective, $\overline{\glorule}_{\Z_G}$ and $\widetilde{\glorule}_{\Z_G}$ are well-defined. 
% \todo{stesso commento del Lemma \ref{f_sur_preserva_Z}, non deriva tutto in automatico dalla fully invariance del centro? perché dobbiamo tirare in ballo la suriettività di $F$?}
\begin{corollary}
 Let $G$ be a finite group and let $\glorule$ be a surjective GCA on $G.$ It holds that $\glorule$ is injective if and only if  $\overline{\glorule}_{\Z_G}$ is injective.    
\end{corollary}
\begin{proof}
By Lemma \ref{ker_sta__nel_centro_se_sur}, it follows that $\Ker(\glorule)\subseteq Z_G^{\mathbb Z}.$ 
Hence, it holds that $\Ker(\glorule)=\Ker(\overline{\glorule}).$
The thesis follows from the fact that a GCA is injective if and only if its kernel contains only $e^{\mathbb Z}.$
\end{proof}

\begin{corollary}
 Let $G$ be a finite group and let $\glorule$ be a surjective GCA on $G.$ It holds that  $\widetilde{\glorule}_{\Z_G}$ is bijective.   \end{corollary}
\begin{proof}
Let $[c]$ be an element of $\Ker(\widetilde{\glorule}),$ that is to say
$\glorule(c)\in \Z_{G}^{\mathbb Z}.$
Since $\glorule$ is surjective, by Theorem~\ref{surj_inj}, $\overline{\glorule}_{\Z_{G}}$ is also surjective.
Thus, there exists $c'\in\Z_{G}^{\mathbb Z}$ such that $\glorule(c')=\glorule(c).$
Hence, $c=c'z$ for some $z\in \Ker(\glorule).$
Since, by Lemma~\ref{ker_sta__nel_centro_se_sur}$, \Ker(\glorule)\subseteq \Z_{G}^{\mathbb Z}$,  it follows that $c\in \Z_{G}^{\mathbb Z}$ and so $[c]=[e].$ As a consequence, we get that $\Ker(\widetilde{\glorule})=\{[e]\}$ and the thesis easily follows.
\end{proof}

If
$\glorule$ is a surjective GCA on a group $G$, then its local rule
$\locrule$ is clearly surjective.
One may wonder whether the converse holds.
The answer is negative, as the following example shows.

\begin{example}\label{f_sur_F_no}
Let $S$ be a non-trivial group and consider the group $G=S\times S.$ Consider the following two endomorphisms $h_{-1}$ and $h_{1}$ of $G$ defined by $h_{-1}(x,y):=(y,e)$ and $h_1(x,y)=(e,y)$ and let $h_0$ be the trivial endomorphism $h_0(x,y)=e.$ Notice that the images of these  endomorphisms commute element-wise so their product defines an homomorphism $f\in \Hom(G^3,G)$ such that $f=(h_{-1},h_0,h_1).$ Consider a GCA $\glorule$ on $G$ with local rule $f.$ Notice that $\Ker(h_1)=\Ker(h_2)=S\times \{e\}.$ As a consequence, any configuration $c\in G^{\mathbb Z}$ such that $c_i=(a_i,e)$ with $a_i\in S$ belongs to the kernel of $\glorule.$ Thus $\glorule$ is not surjective. On the other hand, the homomorphism $f$ is clearly surjective since $f((a_{-1},b_{-1}),(a_0,b_0),(a_1,b_1))=(b_{-1},b_1).$
\end{example}
% Example \ref{f_sur_F_no} shows that there exists non-surjective GCA with surjective local rule. In fact, according to a well-known result by Hedlund \cite{hedlund69}, any CA $\glorule$ is surjective if and only if all functions $f_{(n)}$ obtained by a repeated shifted application of the local rule $f$  are surjective. 
% \todo{
% Luciano: non capisco cosa vuoi dire nel paragrafo di sopra e anche la shifted application non è chiara
% }
The following proposition ensures that for GCA surjectivity  is equivalent to both openness and DPO.

\begin{proposition}\label{surj_open}
Let $\glorule$ be a GCA on a finite group $G$. Then, the following statements are equivalent:\\
$(1)$ $\glorule$ is surjective;\\
$(2)$ $\glorule$ is open;\\
$(3)$ $\glorule$ has DPO.
\end{proposition}
%    \begin{enumerate}
%    \item $\glorule$ is surjective,
%    \item $\glorule$ is open, and
 %   \item $\glorule$ has DPO.
%\end{enumerate}
\begin{proof}
It is well known that a surjective homomorphism of Polish group is open \cite[Thm. 1.5]{hofmann2007open}.
On the other hand, for any ($1$-dimensional) CA $\glorule$, the openness of $\glorule$ implies that $\glorule$ has DPO which, in turn, implies the surjectivity of $\glorule$  \cite{DennunzioFGM2020TCS}.
\end{proof}

We now turn our attention to several other dynamical properties of GCA. As in the case of injectivity and surjectivity, our goal is to reduce their study to that of GCA defined on smaller groups.

Before proceeding, however, we prove 
the equivalence of various forms of chaotic behavior of a dynamical system in the particular case of GCA. For the definition of \textit{total transitivity} and \textit{topologically weakly mixing} and \textit{topologically mixing} map, we address  the reader to~\cite{Moothathu2005}, while for the definition of \textit{ergodically strongly mixing}, \textit{ergodically weakly  mixing} and \textit{ergodic} map, see \cite{aoki1994topological}. Recall that in our case these definitions refer to the Haar measure and to the prodiscrete topology on the group $G^{\mathbb Z}.$  

\begin{theorem}\label{equiv_prop}
Let $\glorule$ a GCA on a finite group $G$. The following statements are equivalent:\\
$(1)$ $\glorule$ is topologically transitive;\\
$(2)$ $\glorule$ is totally transitive;\\
$(3)$ $\glorule$ is topologically weakly  mixing;\\
$(4)$ $\glorule$ is topologically mixing;\\
$(5)$ $\glorule$ is ergodically weakly mixing;\\
$(6)$ $\glorule$ is ergodically strongly mixing;\\
$(7)$ $\glorule$ is ergodic. 
\end{theorem}
\begin{proof}
The following chains of implications are true: \\
% \\
% $(1)\iff (2) \iff (3)$, $(5)\iff (6) \iff (7)$, and $(1) \iff (7)$.
$(1)\iff (2) \iff (3)$: holds for general CA \cite[Prop. 9.2]{Moothathu2005}.\\
$(6)\implies (4) \implies (1)$: holds for general CA   \cite{mfcs19}.\\
$(5)\iff (6) \iff (7)$: follows from the fact that they are equivalent for continuous endomorphisms of compact groups  \cite[Thm. 2]{Chu} and the fact that a global rule of any GCA on $G$ is a continuous endomorphism of the compact group $G^{\mathbb Z}$ (equipped with the prodiscrete topology).\\
$(1) \iff (7)$: follows from the fact that a surjective endomorphism (and more generally an affine transformation) of a compact metric group is ergodic (with respect to the Haar measure) if and only if it is topologically transitive  \cite[Cor. 5.5]{walters1975ergodic}. 
\end{proof}
Theorem~\ref{equiv_prop} allows us, in what follows, to consider only topological transitivity, as the other properties listed above are equivalent to it.

% Now we present a condition that is equivalent to the topological transitivity for GCA and that will be useful in the following. 

% \begin{lemma}\label{lemma_equiv_trans}
%     Let $\glorule$ a GCA over a group $G$. Then, $\glorule$ is topologically transitive if and only if, for every pair of nonempty cylinders $U$ and $V$ in $G^{\mathbb Z}$ and every pair of nonempty cylinders $I$ and $I'$ neighborhoods of $e^{\mathbb Z}$ in $G^{\mathbb Z}$, there exist integers $n,m\geq 0$ such that $\glorule^n(U)\cap I\neq \emptyset$ and $\glorule^m(I')\cap V\neq \emptyset$.
% \end{lemma}
% \begin{proof}
% If $\glorule$ is topologically transitive, then, by definition, for every pair of open nonempty sets $U$ and $V$ in $G^{\mathbb Z}$ there exists an integer $n\geq 0$ such that $\glorule^n(U)\cap V\neq \emptyset$. This implies the second condition of the lemma.

% Suppose now, given two nonempty cylinder $U,V$ in $G$, that the second condition of the lemma is true. Thus, given a cylinder $I$ containing $e^{\mathbb Z}$, there exists $n\geq 0$ with $\glorule^n(U)\cap I\neq \emptyset$. 

% This implies that there is an open set $U'\subseteq U$ such that $\glorule^n(U')\subseteq I$. Hence $\glorule^n(U')$ contains an open neighborhood $I'$ of $e^{\mathbb Z}$ and, by our assumption, there exists $m\geq 0$ such that $\glorule^m(I')\cap V\neq \emptyset$.
% This implies that $\glorule^{n+m}(U)\cap V\neq \emptyset$ and hence the transitivity of $\glorule$. 
% \end{proof}

\begin{theorem}\label{transitive}
Let $G$ be a finite group and $H\trianglelefteq G$. Let $\glorule$ be a GCA on $G$ such that $\glorule(H^{\mathbb Z})\subseteq H^{\mathbb Z}$. 
If $\widetilde{\glorule}$ and $\overline{\glorule}$ are topologically transitive, then $\glorule$ is topologically transitive. 
Moreover, if $\glorule$ is topologically transitive, then $\widetilde{\glorule}$ is topologically transitive. 
\end{theorem}
\begin{proof}

% The fact that $\widetilde{\glorule}$ and $\overline{\glorule}$ are continuous endomorphisms of $G/H$ and $H$ is immediate. 

% \todo{richiamare i fatti seguenti?

% 1: Let $f:X\to Y$
%  be continuous function, $A\subseteq X$
%  equipped with subspace topology, then $f|A:A\to  Y$
%  is continuous.

% 2: Given a continuous map $F:X\to Y$
%  which descends to the quotient, the corresponding map $F:X/H\to Y$
%  is continuous with respect to the quotient topology.

% 3:F descends to the quotient if and only if F(x)
%  depends on x  only through its equivalence class
% }

Let $\cyl{G}$ be the clopen basis of cylinders for $G^{\mathbb Z}$. The corresponding clopen basis for $H^{\mathbb Z}$ is $\cyl{H}=\cyl{G}\cap H^{\mathbb Z}$.
Let $\pi$ be the projection map $\pi:G\to G/H$ defined by 
$\pi(g)=gH=[g]$. With  slight abuse of notation we will use $\pi$ to denote also the corresponding projection from $G^{\ZZ}$ to $(G/H)^{\ZZ}.$

%In the following we will denote by $\pi$ the projection map $\pi:G\to G/H$ that associate each element $g\in G$ to its coset $[g].$

% Suppose that $\glorule$ is a topologically transitive GCA over $G$. We want to show that $\overline{\glorule}$ is topologically transitive.

% To this aim it is sufficient to show that for any $A, B\in \mathfrak{Cyl}_G$ with $A\cap H^{\mathbb Z}\neq \emptyset$ and $B\cap H^{\mathbb Z}\neq \emptyset$, there exists $n$ such that $\glorule^n(A\cap H^{\mathbb Z})\cap B\cap H^{\mathbb Z}\neq \emptyset$.
% Notice that such $A$ and $B$ are cylinders in $G^{\mathbb Z}$ that are neighborhoods of elements in $H^{\mathbb Z}$

% Since $\glorule$ is topologically transitive, there exists an integer $n\geq 0$ such that $\glorule^n(A)\cap B\neq \emptyset$. 

% The map $\glorule$ is continuous, thus there exists  $A'\subseteq A$ open in $G$ such that $\glorule^n(A')\subseteq B$.

% Notice that every cylinder neighborhood of an element in $H^{\mathbb Z}$ contains infinite elements in $H^{\mathbb Z}$.  
% Let $h\in A'\cap H^{\mathbb Z}$ be such an element.  Since $\glorule(H^{\mathbb Z})\subseteq H^{\mathbb Z}$, $\glorule^n(h)\in H^{\mathbb Z}$.
% Thus $\glorule^n(h)\in H^{\mathbb Z}\cap B$.

% This shows that $\overline{\glorule}$ is topologically transitive on $H$.

First, we show that if $\glorule$ is topologically transitive, then $\widetilde{\glorule}$ is also topologically transitive. 
Let $\widetilde{U}$ and $\widetilde{V}$ be two nonempty open sets in $(G/H)^{\mathbb Z}$. Consider the open sets $\pi^{-1}(\widetilde{U})$ and 
$\pi^{-1}(\widetilde{V})$. There exists $n\geq 0$ and $x\in \pi^{-1}(\widetilde{U})$ such that $\glorule^n(x)\in \pi^{-1}(\widetilde{V})$.
Hence, $\widetilde{\glorule}^n([x])=[\glorule^n(x)]\in V$ with $[x]\in U$ and the map $\widetilde{\glorule}$ is topologically transitive. 

Now we prove that, if $\widetilde{\glorule}$ and $\overline{\glorule}$ are topologically transitive, then $\glorule$ is topologically transitive as well. We exploit the following 
well-known result by Moothathu~\cite{Mo05}:  if $F$ and $G$ are any two topologically transitive CA over the alphabets $A$ and $B$, respectively, then the product CA $F\times G$ over the alphabet $A\times B$ is also topologically transitive. We also use the following fact: if a CA $F$ over the alphabet $A$ is topologically transitive, then for every $k>0$ and for every pair of integers $i,j,$  there exists $n_k$ such that, for every $X,Y\in A^k$  the condition  $F^{n_k}(C([i,i+k-1],X))\cap C([j,j+k-1],Y)\neq \emptyset$ holds. 
Hence, by Moothathu's result, if $F$ and $G$ are two topologically transitive CA over the alphabets $A$ and $B,$ for every $k>0$ and for every pair of integers $i,j,$ there exists $n_k$ such that, for every $X,Y\in A^k$ and $Z,W\in B^k,$ both the conditions $F^{n_k}(C([i,i+k-1],X))\cap C([j,j+k-1],Y)\neq \emptyset$ and $G^{n_k}(C([i,i+k-1],Z))\cap C([j,j+k-1],W)\neq \emptyset$ hold. 

Now consider the case $F=\overline{\glorule}$ and $G=\widetilde{\glorule}$. Assume that both these GCA are topologically transitive. We are going to show that  $\glorule$ is also topologically transitive.  

Let $i,j,k$ be any three integers, with $k>0$.  Consider any two words $u=g_1g_{2}\cdots g_{k}, u'=g'_1\cdots g'_{k}\in G^k$ and let $[u]=[g_1]\cdots[g_{k}]$ and $[u']=[g'_1]\cdots [g'_{k}]$ be the corresponding words in $(G/H)^k$.
Let $n_k$ be the positive integer from Moothathu's result above mentioned.
Thus, it holds that $\widetilde{\glorule}^{n_k}(C([i,i+k-1],[u]))\cap C([j,j+k-1],[u'])\neq \emptyset$, where both $C([i,i+k-1],[u])$ and $C([j,j+k-1],[u'])$
are cylinders in $(G/H)^{\mathbb Z}.$
This means that there exists a configuration $c$ in $(G/H)^{\mathbb Z}$ with  $c\in C([i,i+k-1],[u])$ such that $\widetilde{\glorule}^{n_k}(c)=c'$ where $c'\in C([j,j+k-1],[u']).$

 Equivalently, if  $g$ and $g'$ are any two configurations in $G^{\mathbb Z}$ such that $g\in C([i,i+k-1],u)$, $g'\in C([j,j+k-1],u')$, $[g]=c$, and $[g']=c'$, we get $\glorule^{n_k}(g)=g'h'$ for some $h'\in H^{\mathbb Z}.$ Denote by $h'_i$ the $i$-th element of $h'.$
Now consider the words $e^k$ and ${h'_1}^{-1}\cdots {h'_k}^{-1}$ in $H^k.$
Again by Moothathu's result, $\overline{\glorule}^{n_k}(C([i,i+k-1],e^k))\cap C([j,j+k-1],h'_1\cdots h'_k)\neq \emptyset$, where the two are cylinders appearing in this intersection condition are subsets of $H^{\mathbb Z}$. 
This means that there exists a configuration $h\in C([i,i+k-1],e^k)$ such that $\overline{\glorule}^{n_k}(h)=h''$ for some $h''\in C([j,j+k-1],{h'_1}^{-1}\cdots {h'_k}^{-1}).$
Now consider $\glorule^{n_k}(gh)$. It holds that
$\glorule^{n_k}(gh)=\glorule^{n_k}(g)\glorule^{n_k}(h)=g'h'h''.$  
This shows that $$\glorule^{n_k}(C([i,i+k-1],u))\cap C([j,j+k-1],u')\neq \emptyset\enspace,$$ and therefore $\glorule$ is topologically transitive. 
\end{proof}
%We conjecture that in fact the transitivity of $\glorule$ is equivalent to the transitivity of $\overline{\glorule}$ and %of $\widetilde{\glorule}.$
We leave the following question open, as we have not been able to answer it.
\begin{question}\label{conjecture_trans}
Let $G$ be a finite group and $H\trianglelefteq G$. Let $\glorule$ be a topologically transitive GCA on $G$ such that $\glorule(H^{\mathbb Z})\subseteq H^{\mathbb Z}$. 
Is $\overline{\glorule}$ topologically transitive ?
\end{question}

%We have not been able to prove that if $\glorule$ is topologically transitive then $\overline{\glorule}$ is topologically transitive as well.
A positive answer to Question \ref{conjecture_trans}  would allow us to conclude that 
$\glorule$
 is topologically transitive if and only if 
$\overline{\glorule}$
  and 
$\widetilde{\glorule}$
  are topologically transitive.

%\begin{conjecture}\label{conjecture_trans}
%Let $G$ be a finite group and $H\trianglelefteq G$. Let %$\glorule$ be a GCA on $G$ such that $\glorule(H^{\mathbb Z})\subseteq H^{\mathbb Z}$. \end{conjecture}
%Then, $\glorule$ is topologically transitive if and only if $\widetilde{\glorule}$ and $\overline{\glorule}$ are topologically transitive.  

%In what follows, we consider two widely studied dynamical properties that are more restrictive than topological transitivity: namely, strong transitivity and positive expansivity proving a somewhat surprising result: no GCA satisfies these properties.

\bigskip
In what follows, we consider two widely studied dynamical properties that are more restrictive than topological transitivity—namely, strong transitivity and positive expansivity—and prove that no GCA on a finite non-abelian group satisfies these properties.

\begin{theorem}\label{str_trans_su_nonabel_nonesiste}
Let $\glorule$ be a GCA on a finite non-abelian group $G$. It holds that  $\glorule$ is neither strongly transitive nor positively expansive. 
\end{theorem}
\begin{proof}
Assume that there exists a strongly transitive GCA $\glorule$ on a finite non-abelian group $G$. 
Since $\glorule$ is strongly transitive, $\glorule^m$ is surjective for every $m\in \mathbb N.$ By Lemma \ref{ker_sta__nel_centro_se_sur}, we have that $\Ker(\glorule^m)\subseteq \Z_G^{\mathbb Z}$ for every $m\in \mathbb N.$

Now consider a cylinder $U$ containing only configurations composed by elements of $G^{\mathbb Z}\setminus Z_G^{\mathbb Z}$ as, for example, the cylinder containing all the configurations $c$ such that $c_0$ is a fixed element which does not belong to $\Z_G.$ Since  $\Ker(\glorule^m)\subseteq \Z_G^{\mathbb Z},$ we get $e^{\mathbb Z}\not\in \bigcup_{m\in \mathbb N} \glorule^m(U).$ 
This contradicts the assumption that $\glorule$ is strongly transitive.
The second assertion follows from the fact that any positively expansive CA is strongly transitive. 
\end{proof}

The previous theorem does not come as a surprise. In fact, similar results hold in other contexts. 
As an example, it is known that if a compact connected topological group admits a positively expansive endomorphism, it must be abelian \cite{Lam}.

%Now we consider sensitivity to initial conditions. 
We now proceed to consider another well-studied dynamical property that is less restrictive than topological transitivity, namely sensitivity to initial conditions.

First of all, we prove a technical result about the structure of the local rule of the power of a given GCA. It will be useful in the sequel. 

% \begin{lemma}\label{local_rule_of_powers}

% Let $\glorule$ be a GCA over a group $G.$ For every $n\geq 1,$ let  $f^{(n)}=(h^{(n)}_1,h^{(n)}_2,\ldots h^{(n)}_{k_n})$ be the local of the GCA $\glorule^n$ and let $v^{(n)}=\{v^{(n)}_1,v^{(n)}_2,\ldots,v^{(n)}_{k_n}\}$ be its neighborhood vector. 

%  Then we have $$h^{(n)}_j(x)=\prod_{\{(i_1,i_2,\ldots, i_n)\in (\mathbb Z_+)^n \,:\, v^{(1)}_{i_1}+v^{(1)}_{i_2}+\ldots+v^{(1)}_{i_n}=v^{(n)}_j\}}h^{(1)}_{i_n}(h^{(1)}_{i_{n-1}}(\ldots h^{(1)}_{i_{1}}(x)\ldots))$$ for every $x\in G$ and for every $1\leq j\leq k_n.$ 

%   \end{lemma}

\begin{lemma}\label{local_rule_of_powers}
Let $\glorule$ be a GCA over a group $G.$ For every $n\geq 1,$ let $\rho_n$ be the radius of the GCA $\glorule^n$ and let $f^{(n)}=(h^{(n)}_{-\rho_n},\ldots, h^{(n)}_{\rho_n})$ be its local rule.
 % Then $\rho_n\leq n\rho_1$ and   $$h^{(n)}_j(x)=\prod_{\{(i_1,i_2,\ldots, i_n)\in 
 % \ZZ^n:\, i_1+i_2+\cdots+i_n=j\}}h^{(1)}_{i_n}(h^{(1)}_{i_{n-1}}(\ldots h^{(1)}_{i_{1}}(x)\ldots))$$ for every $x\in G$ and for every $-\rho_n \leq j\leq \rho_n.$ 
 It holds that $\rho_n\leq n\rho_1$ and 
 $$h^{(n)}_j(x)=\prod_{\substack{(i_1,i_2,\ldots,i_n)\in \{-\rho_1,\ldots, \rho_1\}^n:\\ i_1 + i_2 + \cdots + i_n = j}}h^{(1)}_{i_n}(h^{(1)}_{i_{n-1}}(\cdots h^{(1)}_{i_{1}}(x)\cdots))$$ for every $x\in G$ and for every $-\rho_n \leq j\leq \rho_n.$ 
\end{lemma}
%ho rimpiazzato \ZZ^n con \{-\rho_1,\ldots, \rho_1\}^n
\begin{proof}
The proof immediately follows  by the definition of GCA. Notice that the radius of $\glorule^n$ can be smaller than $n\rho_1,$ where $\rho_1$ is the radius of $\glorule,$ since the endomorphisms $h^{(n)}_{-n\rho_1}$ and $h^{(n)}_{n\rho_1}$ could be trivial. 
\end{proof}
The following lemma (proved also in \cite{BeaurK24} in a different setting and using a different technique) guarantees that a GCA is either equicontinuous or sensitive to initial conditions. 
%\todo{lo fa anche Kari, qui è più elegante}
\begin{lemma}\label{equic_or_sens}
Let $\glorule$ a GCA over a finite group $G$. Then, $\glorule$ is either equicontinuous or sensitive to initial conditions.
\end{lemma}
\begin{proof}
The assertion is true for continuous endomorphisms of completely
metrizable groups \cite[Thm. 3.11]{JIANG2025129033}.
\end{proof}
The following result establishes a condition equivalent to sensitivity to initial conditions over metric groups. 

\begin{lemma}\label{lemma_sensitive_equiv}
Let $\mathcal G$ be a metric group and $\phi$ an endomorphism of $\mathcal G$. Then, $\phi$ is sensitive to initial conditions if and only if there exists $\epsilon>0$ such that, for any $\delta>0$ there is an element $g\in \mathcal G$ and an integer $t\geq 0$ such that $0<d(g,e)<\delta$ and $d(\phi^t(g),e)\geq\epsilon$.
\end{lemma}
\begin{proof}
If $\phi$ is sensitive  to initial conditions the asserted condition is trivially satisfied. Suppose now that the condition holds. We want to prove that $\phi$ is sensitive, i.e., that there exists $\epsilon>0$ such that, for any $\delta>0$ and any configuration $g'\in \mathcal G$, there is a configuration $g''\in \mathcal G$ and an integer $t\geq 0$ such that $0<d(g'',g')<\delta$ and $d(\phi^t(g''),\phi^t(g'))\geq\epsilon$.

Fix $\epsilon>0.$ Since the condition holds, for every $\delta>0$ there exist $g\in \mathcal{G}$ and an integer $t\geq 0$ with $0<d(g,e)<\delta$ and $d(\phi^t(g),e)\geq \epsilon.$

Given any element $g'$ we get 
 $d(gg',g')=d(g,e)$ and $d(\phi^t(gg'),\phi^t(g'))=d(\phi^t(g)\phi^t(g'),\phi^t(g'))$$=d(\phi^t(g),\phi^t(e))$. Thus we can take $g''=gg'.$
\end{proof}

We now prove that sensitivity to initial conditions for arbitrary GCA $\glorule$ can be characterized in terms of the radius of the iterates $\glorule^n$ of $\glorule$. This generalizes the case of GCA over $\mathbb Z_m$ \cite[Lemma 4.1]{ManziniM99}.

\begin{lemma}\label{radius_sensitive}
 Let $\glorule$ be a GCA on the group $G$. It holds that $\glorule$ is sensitive  to initial conditions if and only if $\limsup_{n\to \infty}\rho(\glorule^n)=\infty.$
 Equivalently, $\glorule$ is equicontinuous if and only if $\limsup_{n\to \infty}\rho(\glorule^n)<\infty.$  
 \end{lemma}

 \begin{proof}
 If $\limsup_{n\to \infty}\rho(\glorule^n)$ is finite then there exists $k>0$ such that $\rho(\glorule^n)<k$ for every $n$. Then, for every $\epsilon>0$, if $g\in G^{\mathbb Z}$ is any configuration sufficiently close to $e^{\mathbb Z}$, it holds that $d(\glorule^n(g),e^{\mathbb Z})<\epsilon$ and, hence, by Lemma~\ref{lemma_sensitive_equiv},   $\glorule$ is not sensitive  to initial conditions. 

Suppose now that $\limsup_{n\to \infty}\rho(\glorule^n)=\infty$,
i.e., for every $k$ we can find $n$ such that $\rho(\glorule^n)>k$. 
% Thus $\rho(\glorule^n)=|\min(v)|$ or $|\max(v)|$, where $v$ is the neighbor vector of the local rule $\locrule^{(n)}$ of $\glorule^n$.
Set $\rho(\glorule^n)=\rho_n$ and let $f^{(n)}=(h_{-\rho_n}^{(n)},\ldots, h_{\rho_n}^{(n)})$ be the local rule of $\glorule^n.$
Without loss of generality we can assume that $h^{(n)}_{-\rho_n}$ is a non-trivial endomorphism. 
% $\rho(\glorule^n)=|\min(v)|$ and denote $\min(v)$ by $j$.
% Denote by $h^{(n)}_1$ the first component of the local rule $\locrule^{(n)}$ of $\glorule^n$. This is the endomorphism corresponding to the element $j$ of the neighbor vector and hence  $h^{(n)}_1$ is different from the trivial endomorphism.
Consider a configuration $g\in G^{\mathbb Z}$ such that $g_{-\rho_n}\notin \Ker(h^{(n)}_{-\rho_n})$ and $g_i=e$ for every $i\neq -\rho_n$.
Clearly $d(g,e^{\mathbb Z})=\frac{1}{2^{\rho_n}}$
and $\glorule^n(g)(0)=h^{(n)}_{-\rho_n}(g_{-\rho_n})\neq e$. Hence, $d(\glorule^n(g),e^{\mathbb Z})=1$. By Lemma~\ref{lemma_sensitive_equiv}, this shows that $\glorule$ is sensitive  to initial conditions with $\epsilon=1$. 
% Consider a configuration $g\in G^{\mathbb Z}$ such that $g_j\notin \Ker(h^{(n)}_1)$ and $g_i=e$ for every $i\neq j$.
% Clearly $d(g,e^{\mathbb Z})=\frac{1}{2^{|j|}}$
% and $\glorule^n(g)(0)=h^{(n)}_1(g_j)\neq e$. Hence $d(\glorule^n(g),e^{\mathbb Z})=1$. This shows that $\glorule$ is sensitive with $\epsilon=1$.
\end{proof}
% \todo{dobbiamo uniformare la notazione per l'iterata di una mappa. o si usa $f^{(n)}$ o si usa $f^n$. questo sia per $f$ che per $h$ che per $\glorule$}
\begin{remark}\label{observ_equic}
 A trivial consequence of Lemma \ref{radius_sensitive} is the following. With the same notations used in Lemma \ref{local_rule_of_powers}, we can state that a GCA $\glorule$ is equicontinuous if and only if there exists a positive constant $K$ such that the endomorphisms $h^{(n)}_j$ are trivial for every $n\geq 1$ and for every $j$ such that 
 $|j|>K.$
 %$|v^{(n)}_j|>K.$
\end{remark}

\begin{theorem}\label{sensitivity}
Let $G$ be a finite group and $H\trianglelefteq G$. Let $\glorule$ be a GCA on $G$ such that $\glorule(H^{\mathbb Z})\subseteq H^{\mathbb Z}$. It holds that $\glorule$ is equicontinuous if and only if both $\widetilde{\glorule}$ and $\overline{\glorule}$ are equicontinuous. Equivalently,   $\glorule$ is sensitive to initial conditions  if and only if either $\widetilde{\glorule}$ or $\overline{\glorule}$ is sensitive to initial conditions.   
\end{theorem}
\begin{proof}
Denote by $f^{(n)}=(h^{(n)}_{-\rho_n},\ldots, h^{(n)}_{\rho_n})$, $\widetilde f^{(n)}=(\widetilde h^{(n)}_{-\rho_n},\ldots, \widetilde h^{(n)}_{\rho_n})$ and $\overline f^{(n)}=(\overline h^{(n)}_{-\rho_n},\ldots, \overline h^{(n)}_{\rho_n})$ the local rules of $\glorule^n$, $\widetilde{\glorule
 }^n$ and   $\overline{\glorule}^n,$ respectively. We assume, as usual, that, for every $n\geq1,$ at least one between the endomorphisms $h^{(n)}_{-\rho_n}$ and $h^{(n)}_{\rho_n}$ is non-trivial. However, the radius of $\overline{\glorule}^n$ ($\widetilde{\glorule}^n,$ respectively) could be smaller than $\rho_n$ since $\overline{h}^{(n)}_{-\rho_n}$ and $\overline{h}^{(n)}_{\rho_n}$ ($\widetilde{h}^{(n)}_{-\rho_n}$ and $\widetilde{h}^{(n)}_{\rho_n},$ respectively) could be both trivial.

Hence,  we get that  $\rho(\widetilde{\glorule}^n)\leq \rho(\glorule^n)$ and $\rho(\overline{\glorule}^n)\leq \rho(\glorule^n)$ for every $n$.
Thus, by  Lemma \ref{radius_sensitive}, 
 if at least one between $\widetilde{\glorule}$ and $\overline{\glorule}$ is sensitive to initial conditions,   $\glorule$ is also sensitive to initial conditions. Equivalently, if $\glorule$ is equicontinuous,  $\widetilde{\glorule
 }$ and $\overline{\glorule}$ are also equicontinuous. 
%Now we show that the converse is also true. 

We now prove the converse implication: if both $\overline{\glorule}$ and $\widetilde{\glorule}$ are equicontinuous, then $\glorule$ is, too.  So, assume that $\widetilde{\glorule
 }$ and $\overline{\glorule}$ are equicontinuous. By Remark~\ref{observ_equic}, there exist
 \begin{itemize}
     \item a constant $\widetilde{K}$ such that every $G/H$-endomorphism $\widetilde h^{(n)}_j$  is trivial for every $n\geq 1$ and every $j$ such that $| j|>\widetilde{K},$ and
     \item a constant $\overline{K}$ such that every $H$-endomorphism $\overline h^{(n)}_j$ is trivial for every $n\geq 1$ and every $j$ such that $|j|>\overline{K}.$
 \end{itemize}
  Set $K=\widetilde{K}+\overline{K}+2\rho_1$. 
  
We want to show that every endomorphism $h^{(n)}_j$ is trivial for every $n\geq 1$ and every $j$ such that $|j|>K.$
To this aim recall that, by Lemma \ref{local_rule_of_powers}, $$h^{(n)}_j(x)=\prod_{\substack{(i_1,i_2,\ldots,i_n)\in \{-\rho_1,\ldots, \rho_1\}^n \\ i_1 + i_2 + \cdots + i_n = j}}h^{(1)}_{i_n}(h^{(1)}_{i_{n-1}}(\cdots h^{(1)}_{i_{1}}(x)\cdots))$$ for every $x\in G$ and for every $-\rho_n\leq j\leq \rho_n.$ 

 Assume that $|j|>K.$
 Note that, given our definition of $K,$ if $(i_1,\dots ,i_n)\in \{-\rho_1,\ldots, \rho_1\}^n$ is such ${i_1}+\dots +{i_n}=j$, then the set 
 \begin{align*}
 A_{(i_1,\ldots ,i_n)}&:=\{(s,\ell)\in \ZZ_+\times\ZZ \,:\, \\ &\qquad {i_1}+\cdots +{i_s}={\ell},\,|{\ell}|>\widetilde{K}\,\mbox{ and }\, |j-\ell|>\overline{K} \}
 \end{align*}
 is non-empty and there exists an element $(s,\ell)$ in this set with minimal value of $s.$ Denote by $(\widehat{s},\widehat{\ell})_{i_1,\ldots,i_{n}}$ such an element.

Consider now the following set of pairs $(\widehat{s},\widehat{\ell})_{(i_1,\ldots,i_n)}$ when $(i_1,\ldots,i_n)$ varies among the $n$-tuples in $\{-\rho_1,\ldots, \rho_1\}$ such that ${i_1}+\cdots+{i_n}=j$:
\begin{align*}
\mathcal{S}_{j,n}&:=\{(\widehat{s},\widehat{\ell})_{(i_1,\cdots,i_n)}\,:\\ &\quad (i_1,\ldots,i_n)\in\{-\rho_1,\ldots, \rho_1\}^n\,\mbox{ s.t. }\, {i_1}+\cdots+{i_n}=j \}\enspace.
\end{align*}
Thus, we can write $h^{(n)}_j$ as
$$h^{(n)}_j(x)=\prod_{(\widehat{s},\widehat{\ell})\in\mathcal S_{j,n}}\prod_{\mathcal B_{\widehat s,\widehat{\ell}}}\prod_{\mathcal C_{\widehat s,\widehat{\ell}}}h^{(1)}_{i_n}(h^{(1)}_{i_{n-1}}(\cdots h^{(1)}_{i_{1}}(x)\cdots))\enspace,$$
where the second product is over the set 
 \begin{align*}
 \mathcal B_{\widehat s,\widehat{\ell}}&:=\{(i_1,i_2,\ldots, i_{\widehat{s}})\in {\{-\rho_1,\ldots, \rho_1\}}^{\widehat s} \,:\\
 &\qquad {i_1}+{i_2}+\cdots+{i_{\widehat s}}={\widehat{\ell}}\}
 \end{align*}
 and the third is over the set 
 \begin{align*}
 \mathcal C_{\widehat s,\widehat{\ell}}&:=\{(i_{\widehat{s}+1},\ldots, i_{n})\in \{-\rho_1,\ldots, \rho_1\}^{n-\widehat s} \,:\\
 &\qquad {i_{\widehat{s}+1}}+\cdots+{i_{n}}={j}-\widehat{\ell}\,\}\enspace.
\end{align*} 
The last expression for $h^{(n)}_j$ can be in turn rewritten as $$h^{(n)}_j(x)=\prod_{(\widehat s,\widehat{\ell})\in\mathcal S_{j,n}}h_m^{(n-\widehat s)}(h^{(\widehat{s})}_{\widehat{\ell}}(x))\enspace,$$ where $m=j-{\widehat{\ell}}.$

 Note that, by our previous assumptions, the endomorphisms $\widetilde h^{(\widehat s)}_{\widehat{\ell}}$ and $\overline{h}^{(n-\widehat s)}_m$ are trivial.  Hence, $h^{(\widehat s)}_{\widehat{\ell}}(x)\in H$ for every $x\in G$ and 
 %the endomorphism $\overline{h}^{(n-\widehat s)}_m$ is trivial hence 
 $h_m^{(n-\widehat s)}(h^{(\widehat{s})}_{\widehat{\ell}}(x))=e$ for every $x \in G.$ Therefore, we get that $h^{(n)}_j$ is trivial and this concludes the proof.
 \end{proof}
% We conjecture that also the converse of Theorem \ref{sensitivity} is true.
% \begin{conjecture}\label{conj_sensitive}
% $\glorule$ is sensitive if and only if at least one between $\widetilde{\glorule}$ and $\overline{\glorule}$ is sensitive. 
% \end{conjecture}
%
% \todo{in tutto il lavoro c'è da sostituire cose del tipo  $a+\ldots +b$ con $a+\cdots +b$,
% qualcuno potrebbe essermi scappato}
%
We end this section by considering the topological entropy of a GCA.
\begin{theorem}\label{entropy}
Let $G$ be a finite group and $H\trianglelefteq G$. Let $\glorule$ be a GCA on $G$ such that $\glorule(H^{\mathbb Z})\subseteq H^{\mathbb Z}$.  Then,
$$\entropy{H}(\glorule)=\entropy{H}(\widetilde{\glorule})+\entropy{H}(\overline{\glorule}).$$
\end{theorem}
\begin{proof}
The proposition holds for endomorphisms of compact groups (see \cite[Thm. 2]{Juzvinskii_1971} and \cite{GIORDANO_BRUNO_VIRILI_2017}).
\end{proof}
It is well known that an equicontinuous map over a compact metric space has zero topological entropy \cite[Prop. 2.11]{SUN}.
For GCA we conjecture that the converse is also true.

\begin{conjecture}\label{conj_sensitivity_top_entropy}
Let $\glorule$ be a GCA over a finite group $G$.
Then, $\glorule$ is equicontinuous if and only if its topological entropy is zero. Equivalently, $\glorule$ is sensitive  to initial conditions if and only if its topological entropy is positive. 
\end{conjecture}
%
%A strategy for iterated quotienting
\section{Algorithimic decomposition technique: iterated quotienting}\label{section_spezzatino}
In Section~\ref{general_res} we have shown that the study of a number of dynamical properties of a GCA $\glorule$ can be reduced to that of the same properties of  two GCA, namely $\overline{\glorule}$ and $\widetilde{\glorule}$, defined on smaller groups. However, these smaller groups do not guarantee that the analysis of the dynamical properties of 
$\overline{\glorule}$ and $\widetilde{\glorule}$
will be simpler than that of the original GCA $\glorule$.

In this section, we show that by iterating the decomposition of the original group a finite number of times, and using suitable quotient groups at each step, we eventually obtain a collection of GCA for which the analysis of the relevant properties becomes significantly simpler, if not immediate. To this end, we begin with some definitions and preliminary results.

Let $G$ be a finite group. 
A group word is a formal expression in terms of variables and group operations (multiplication, inverses, and identity). A \textit{verbal subgroup} of a group $G$ is defined as the subgroup generated by all elements obtained by substituting elements of $G$ into the variables of a given set of group words. As an example, the verbal subgroup of a group $G$ generated by the group word $x y x^{-1} y^{-1}$ is the commutator of $G$.
\smallskip

A finite group $G$ will be called \textit{invariantly simple} (\textit{characteristically simple}, \textit{verbally simple}, respectively) if the only fully invariant (characteristic, verbal, respectively) subgroups of $G$ are the trivial subgroup and $G$ itself.

\begin{theorem}\label{equiv_verbally}
For any  finite group $G$, the following statements are equivalent:\\
$(1)$ $G$ is characteristically simple;\\
$(2)$ $G$ is invariantly simple;\\
$(3)$ $G$ is verbally simple; \\
$(4)$ $G$ is the direct product of isomorphic simple groups.
\end{theorem}
\begin{proof}
By definition, it immediately follows that if $G$ is characteristically simple then it is invariantly simple, and that if $G$ is invariantly simple then it is verbally simple. 
It is also well-known that a finite characteristically simple group is the direct product of isomorphic simple groups \cite[Thm. 1.4]{gorenstein2007finite}.
Moreover, a finite verbally simple group is characteristically simple   \cite[Cor. 53.57]{neumann2012varieties}.
% As a consequence any finite invariantly simple group is the direct product of isomorphic simple groups.
On the contrary, it is easy to see that the direct product of isomorphic simple groups is characteristically simple.
\end{proof}
We now introduce a recursive procedure, namely, a divide and conquer one, called \texttt{Decomposition} that, given a GCA $\glorule$ on a finite group $G$, produces a finite set 
$\{(\glorule_1,G_1),\dots,(\glorule_k,G_k)\}$
of GCA $\glorule_i$ on invariantly simple groups $G_i$ such that   
$\glorule$ on $G$ is surjective, injective, equicontinuous  (and maybe topologically transitive) if and only if 
all the $\glorule_i$ on $G_i$ satisfy the same property.
Moreover, the topological entropy of $\glorule$ on $G$ is the sum 
of the topological entropies of all the $\glorule_i$ on $G_i$.
%\vspace{0.5cm}
\begin{algorithm}[h]
 \KwIn{A GCA $(\glorule, G)$}
 \KwOut{A finite set $\{(\glorule_1,G_1),\dots,(\glorule_k,G_k)\}$ of GCA, where each $G_i$ is an invariantly simple group.}
\SetKwFunction{Decomposition}{Decomposition}
\SetKwProg{Fn}{Function}{:}{}
%\SetKwProg{Fn}{Function}{:}{\KwRet}
\Fn{\Decomposition{$\glorule$,$G$}}
{
\If{$G$ is invariantly simple}{
   \KwRet{$\{(\glorule, G)\}$}
 }
Let $H$ be any non-trivial proper fully invariant subgroup of $G$\;
 a  $\gets$ \Decomposition{$\widetilde{\glorule},G/H$}\;
  b  $\gets$ \Decomposition{$\overline{\glorule},H$}\;
  \KwRet{$a\cup b$}
}
\caption{Decomposition}
\end{algorithm}
%\vspace{0.5cm}
%
Note that the output of the procedure \texttt{Decomposition} depends on the sequence of fully invariant subgroups $H$'s chosen at step $4$ (in all the recursive calls). However, 
 any of the possible outputs of \texttt{Decomposition} will work for our purposes. 
The following result is a trivial consequence of Theorems~\ref{surj_inj},~\ref{transitive},~\ref{sensitivity} and~\ref{entropy}. 

%\ref{strong_trans},\ref{pos_exp},

\begin{theorem}\label{spezzatino}
 Let $\glorule$ be a GCA over a finite group $G$. The following statements hold.\\
$(1)$ $\glorule$ is injective (resp., surjective) (resp., equicontinuous) if and only if all the GCA in the set {\tt Decomposition($\glorule,G$)} 
 are injective (resp., surjective) (resp, equicontinuous), while $\glorule$ is sensitive to initial conditions if and only if at least one GCA in that set is sensitive to initial conditions, too.\\
$(2)$ If all of the GCA in the set  {\tt Decomposition($\glorule,G$)} are topologically transitive, then $\glorule$ is topologically transitive. \\ 
$(3)$ The topological entropy  of $\glorule$
is equal to the sum of the topological entropies of the GCA in the set 
 {\tt Decomposition($\glorule,G$)}.
\end{theorem}

We now give an explicit strategy for selecting the subgroups $H$'s 
that will allow us to simplify the problem of deciding a number of dynamical properties of the GCA under consideration.

Let $G$ be a finite group.
Set $G^{(0)}:=G$ and, for every $i\geq 0$, $G^{(i+1)}:=[G^{(i)},G^{(i)}]$.
The series $\{G^{(i)}\}$ is called \textit{ the derived series} of $G$. 
The derived series of $G$ eventually reaches a \textit{perfect group}, i.e., a group equal to its own commutator. Denote such a group by $\widehat{G}$.
The group $G$ is said to be \textit{solvable} if $\widehat{G}=\{e\}$. 
Solvable groups are a central topic in algebra and includes the widely studied class of \textit{nilpotent} groups \cite{Rotman_groups}.
Moreover, every finite group of odd order is solvable by the celebrated Feit-Thompson theorem. 
We recall that the commutator subgroup $[G,G]$ of a group $G$ is fully invariant and thus normal in $G.$ Moreover, the quotient of $G/[G,G]$ is abelian. 
%\vspace{0.5cm}
\begin{algorithm}[h]
 \KwIn{A GCA $(\glorule, G)$}
 \KwOut{A finite set $\{(\glorule_1,G_1),\dots,(\glorule_k,G_k)\}$ of GCA where  $G_i$'s are either abelian or non-abelian invariantly simple groups. If $G$ is solvable then all $G_i$'s are abelian}
\SetKwFunction{ExplicitDecomposition}{Explicit-Decomposition}
\SetKwFunction{Decomposition}{Decomposition}
\SetKwProg{Fn}{Function}{:}{}
%\SetKwProg{Fn}{Function}{:}{\KwRet}
\Fn{\ExplicitDecomposition{$\glorule$,$G$}}
{
\If{$[G,G]=\{ e \}$}{
  \KwRet{$\{(\glorule,G)\}$}
  }
\If{$[G,G]=G$}{
  \KwRet{$\Decomposition (\glorule,G)$}
  }
$H$ $\gets$ $[G,G]$\;
a  $\gets$ \ExplicitDecomposition{$\widetilde{\glorule},G/H$}\;
b  $\gets$ \ExplicitDecomposition{$\overline{\glorule},H$}\;
\KwRet{$a\cup b$}
}
\caption{Explicit Decomposition}
\end{algorithm}
%\vspace{0.5cm}
It is not hard  to verify  that, if $G$ is a solvable finite group, then all the GCA produced by {\tt Explicit-Decomposition($\glorule$,$G$)} are abelian.
Since a complete characterization of surjectivity, injectivity, and sensitivity  to initial conditions for GCA over abelian groups exists 
%\cite{
%DennunzioFMMP19,
%DBLP:journals/access/DennunzioFM23,
%DBLP:journals/isci/DennunzioFM24,
%Dennunzio20JCSS,
%kari2000,
%DennunzioFGM2020INS,
%DennunzioFGM2020TCS
%}
\cite{
Dennunzio20JCSS,kari2000}, we also have a characterization of these properties in the case of solvable groups.
%\todo{fare le citazioni giuste per la sensitivita}

% \todo{come strutturare il seguente corollario? 
% Luciano: forse lo eliminaerei essendoci una intera sezione dopo sulla decidibilità}

% \begin{corollary}
% Injectivity, surjectivity and sensitivity are decidable properties for GCA over finite solvable groups.  
% \end{corollary}
%
%
\smallskip 

We now consider the case of non-solvable groups. In this case, the set {\tt Explicit-Decomposition($\glorule$,$G$)} also contains  GCA on products of isomorphic, non-abelian simple groups.

Let $S$ be a finite non-abelian simple group and let $G$ be the product of $m$ copies of $S$, \ie, $G=S_1\times\cdots\times S_m$, where $S_i\cong S$ for every $i$.
In the following we will identify the subgroup $\{g\in G : g_i\in G \,\mbox{and}\, g_j=e\;\forall j\neq i\}$ of $G$ with $S_i$.

Since the kernel of every endomorphism $h$ of $G$ is a normal subgroup  and a normal subgroup of a direct product of non-abelian simple groups is the direct product of some of them (see, for instance,~\cite[p. 174]{dummit2003abstract}), it follows that $\Ker(h)=\prod_{t\in I} S_t$, where $I$ is a nonempty proper subset of $\{1,\ldots,m\}$.

Consider now a \textit{surjective} GCA $\glorule$ over $G.$ Thus the local rule $\locrule=(h_{-\rho},\ldots,h_{\rho})$ of $\glorule$ is also surjective. 

By Lemma~\ref{surj_normal}, $\Imma(h_i)$ is also a normal subgroup of $G$ and, hence, it holds that $\Imma(h_i)=\prod_{t\in J}S_t$, where $J$ is a nonempty proper subset of $\{1,\ldots,m\}$.

If $i\neq j$, $\Imma(h_j)\subseteq C_G(\Imma(h_i))$. Since $S_t$ is non-abelian, this implies that the factors $S_t$ appearing in $\Imma(h_j)$ are distinct from the factors $S_t$ appearing in $\Imma(h_i)$. Since $\locrule$ is surjective, the $\Imma(h_i)$'s form a partition of the factors $S_t$'s.

If the endomorphism $h_i$ has $r_i$ factors in the image it must have $m-r_i$ factors in the kernel.

Suppose that there exists a simple group $S_r$ among the factors of $G$ which belongs to the factors of $\Ker(h_i)$ for every $i$. In that case the restriction of $\locrule$ to $S_r^k$ induces a GCA over $S_r$ whose local rule is trivial, but this contradicts the fact that $\glorule$ is surjective (see Example~\ref{f_sur_F_no}). 

In particular, the following fact holds.

\begin{remark}\label{FsurKerempty}
If $G$ is a finite group which is the product of simple isomorphic non-abelian groups, then any GCA
$\glorule$ with local rule $f=(h_{-\rho},\ldots,h_{\rho})$ on $G$ is surjective if and only if $$\bigcap_{-\rho\leq j\leq \rho}\Ker(h_j)=\emptyset.$$
\end{remark}

%\todo{questo accade, ad esempio, nell'esempio fatto all'inizio del paper. Dirlo.}

We can conclude that every factor $S_t$ appearing in $G$ is a factor of exactly one of the images of the $h_i$'s and there are no factors $S_t$ which are factors of every kernel of the $h_i$'s. Thus, for every $S_t$, there exists precisely an $i$ such that $S_t$ is not a factor of $\Ker(h_i)$.

In other terms, \textit{the action of the family of  endomorphisms }$h_i$\textit{'s can be thought as a permutation over the set of the factors }$S_t$\textit{ of }$G$. Denote this permutation by $\pi_{\locrule}$ and  by $o$ the order of the permutation $\pi_f$.

We illustrate these facts by an example.

\begin{example}
Consider a local rule of the form $\locrule=(h_{-1},h_0,h_1)$ and a group $G=S_1\times S_2\times S_3\times S_4$. Suppose that $\Imma(h_{-1})=S_2\times S_4$, $\Imma(h_0)=S_3$,  $\Imma(h_1)=S_1$, $\Ker(h_{-1})=S_3\times S_4$, $\Ker(h_0)=S_1\times S_2\times S_3$, and $\Ker(h_1)=S_1\times S_2\times S_4$. We can represent the situation as follows. 

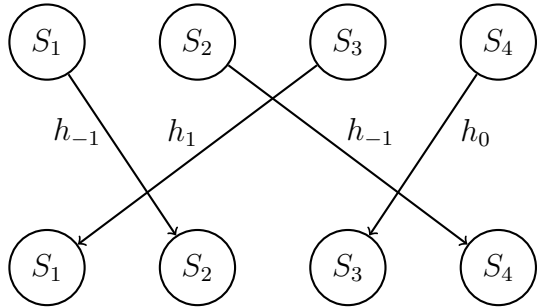
\vspace{0.5cm}
\begin{center}
\begin{tikzpicture}[block/.style={draw,circle, minimum size=1cm, thick, text centered}]
% Blocchi della prima fila
\node[block] (A1) at (0, 3) {$S_1$};
\node[block] (A2) at (2, 3) {$S_2$};
\node[block] (A3) at (4, 3) {$S_3$};
\node[block] (A4) at (6, 3) {$S_4$};
% Blocchi della seconda fila
\node[block] (B1) at (0, 0) {$S_1$};
\node[block] (B2) at (2, 0) {$S_2$};
\node[block] (B3) at (4, 0) {$S_3$};
\node[block] (B4) at (6, 0) {$S_4$};

\node[block, draw=none] (h1) at (0.4, 1.8) {$h_{-1}$};
\node[block, draw=none] (h2) at (1.8, 1.8) {$h_{1}$};
\node[block, draw=none] (h3) at (4.3, 1.8) {$h_{-1}$};
\node[block, draw=none] (h4) at (5.7, 1.8) {$h_0$};

% Frecce della permutazione
\draw[->, thick] (A1) -- (B2); %node[midway, left] {$h_{-1}$};
\draw[->, thick] (A2) -- (B4); %node[midway, left] {$h_{-1}$};
\draw[->, thick] (A3) -- (B1); %node[midway, right] {$h_1$};
\draw[->, thick] (A4) -- (B3); %node[midway, right] {$h_0$}
\end{tikzpicture}
\end{center}
\vspace{0.5cm}

In this case the permutation $\pi_{\locrule}$ is the permutation whose only cycle is $(1,2,4,3)$ and its order is $o=4$.
    
\end{example}

In the following, with a slight abuse of notation, we will identify the endomorphism $h_i$ with the restriction of $h_i$ to those factors of $G$ on which $h_i$ acts non-trivially. This restriction is clearly an automorphism of those factors. 

For any $1\leq i\leq m$, consider $S_i$, the $i$-th factor of $G$. Then, there exists exactly one $h_j$ acting non-trivially on $S_i$ whose image is another $S_{i^{(1)}}$. Over the factor $S_{i^{(1)}}$ there exists exactly one $h_{j^{(1)}}$ acting non-trivially and so on. After $o$ steps we will return on the initial factor $S_i$ applying the endomorphism $h_{j^{(o-1)}}$ to the factor $S_{i^{(o-1)}}$. Define the automorphism $\hat h_i$ of $S_i$ in the following way: 
\[\hat h_i:=h_{j^{(o-1)}} \circ \ldots\circ h_{j^{(1)}}\circ h_j\enspace,
\] 
and denote by $o_i$ the order $\hat h_i$ as an automorphism (i.e., $o_i$ is the smallest positive number such that the composition of $\hat h_i$ with itself $o_i$ times is the identity map on $S_i$).

\begin{example}
Consider the same local rule of the previous example. In this case $\hat h_1=h_1\circ h_0\circ h_{-1}\circ h_{-1}$, $\hat h_2=h_{-1}\circ h_1\circ h_0\circ h_{-1}$,  $\hat h_3=h_0\circ h_{-1}\circ h_{-1}\circ h_1$, and $\hat h_4=h_{-1}\circ h_{-1}\circ h_1\circ h_0$.
\end{example}

\begin{definition}
A group $G$ which is the product $S_1\times\ldots \times S_m$ of finite, non-abelian, isomorphic simple groups $S_i$ is said to be \emph{minimal} with respect to the action of a given GCA $\glorule$ over $G$ if there are no two non-empty, disjoint sets $I,J$ with $I\cup J=\{1,2,\ldots,m\}$, such that $\glorule(\left(\prod_{i\in I}S_i\right)^{\mathbb Z})\subseteq \left(\prod_{i\in I}S_i\right)^{\mathbb Z}$ and $\glorule(\left(\prod_{i\in J}S_i\right)^{\mathbb Z})\subseteq \left(\prod_{i\in J}S_i\right)^{\mathbb Z}.$
\end{definition}

When $\glorule$ is surjective, the group $G$ is minimal with respect to the action of $\glorule$
if and only if the corresponding permutation $\pi_f$ above defined is a single cycle. Notice that if $G$ is not minimal with respect to $\glorule$ it is possible to decompose the dynamics of $\glorule$ into the product of the dynamics of $\glorule$ restricted to its minimal components. In fact each minimal component $H$ satisfies $\glorule(H^{\mathbb Z})\subseteq H^{\mathbb Z}$. So, in this case we can continue the decomposition previously illustrated  one step further. 

\begin{lemma}\label{lemma_product_simples}
Let $G=S_1\times\ldots\times S_m$ be a product of finite, non-abelian, isomorphic simple groups. Let $\glorule$ be a surjective GCA on $G$ with  local rule $\locrule=(h_{-\rho},\ldots,h_{\rho})$. Let $\Imma(h_i)=\prod_{t\in J_i}S_t$, where $J_i$ is a nonempty proper subset of $\{1,\ldots,m\}$ and let $r_i=|J_i|$. Let $o$ and $o_i$, $1\leq i\leq m$, be defined as above. If $G$ is minimal with respect to $\glorule$, then it holds that 
$$\glorule^{o\gcd(o_1,\ldots,o_m)}=\sigma^{-\gcd(o_1,\ldots,o_m){\sum_i ir_i}}.$$  
% if $\sum_iv_ir_i= 0$,  there exists $n$ such that $\glorule^n$ is the identity. If $\sum_iv_ir_i\neq 0$, there exists $n$ such that $\glorule^n$ is a shift like GCA defined by the identity endomorphism. 
\end{lemma}
%\todo{definire lo shift $\sigma$}
\begin{proof}
% Let $\pi_{\locrule}$ be the permutation of the set $\{S_1,\ldots,S_m\}$ associated to the local rule $f$.

% Let $o$ be the order of this permutation. 

Let $c\in G^{\mathbb Z}$ be any configuration. Thus, we can write $c_i=(a^{(i)}_1,\ldots,a^{(i)}_m)$ with $a^{(i)}_j\in S_j$. Consider the configuration $\glorule^o(c)$. The element $\glorule^o(c)_i$ is $$(\hat h_1(a^{(i+\sum jr_j)}_{1}),\ldots,\hat h_m(a^{(i+\sum jr_j)}_{m})).$$
Therefore,  the element $\glorule^{o\gcd(o_1,\ldots,o_m)}(c)_i$ is $$(a^{(i+\gcd(o_1,\ldots,o_m)\sum jr_j)}_{1},\ldots,a^{(i+\gcd(o_1,\ldots,o_m)\sum jr_j)}_{m}).$$
% where each $\hat h_i$ is a composition of $o$ automorphisms obtained restricting the endomorphisms $h_1,\ldots,h_k$   to the factors of $G$ on which they are not trivial. As such, each $\hat h_i$ has finite order $o_i$ (order as an automorphism, i.e., the composition of $h_i$ with itself $o_i$ times is the identity map).   Taking $n=o\gcd(o_1,\ldots,o_k)$ we get the assertion.
\end{proof}

We illustrate the previous proof by an example.

\begin{example}
\label{extable}
Consider the same local rule $f$ of the previous examples. Notice that $r_1=2$, $r_2=1$ and $r_3=1$ (because $\Imma(h_{-1})$ has two factors, while  $\Imma(h_0)$ and $\Imma(h_1)$ have one factor each).
Consider the GCA $\glorule$ having $f$ as local rule.
In this case $\sum_i ir_i=-1$, $n'=\gcd(o_1,o_2,o_3,o_4)$, and $n=on'=4n'$. Table~\ref{tab:placeholder_label} represents  the behavior of $\glorule$.  
\begin{table*}
\caption{The dynamical behavior of the GCA $\glorule$ from Example~\ref{extable}.}
\begin{center}
\begin{tabular}{cc|c|c|c|c}
       \cline{2-6}
\rule{0pt}{2.5ex} $c $  & $\cdots$ & $\cdots$ & $a_1^{(i)}a_2^{(i)}a_3^{(i)}a_4^{(i)}$ & $\cdots$ & $\cdots$ \\ \cline{2-6}
\rule{0pt}{2.5ex} $\glorule(c) $ & $\cdots$ & $\cdots$ & $h_1(a_3^{(i+1)})h_{-1}(a_1^{(i-1)})h_0(a_4^{(i)})h_{-1}(a_2^{(i-1)})$ & $\cdots$ & $\cdots$ \\ \cline{2-6}
\rule{0pt}{2.5ex} $\glorule^2(c) $ &$\cdots$ & $\cdots$ & $\cdots$ & $\cdots$ & $\cdots$ \\ \cline{2-6}
\rule{0pt}{2.5ex} $\glorule^3(c) $ &$\cdots$ & $\cdots$ & $\cdots$ & $\cdots$ & $\cdots$ \\ \cline{2-6}
\rule{0pt}{2.5ex} $\glorule^4(c) $ & $\cdots$ & $\cdots$   & $\hat h_1(a_1^{(i-1)})\hat h_{2}(a_2^{(i-1)})\hat h_3(a_3^{(i-1)})\hat h_{4}(a_4^{(i-1)})$ & $\cdots$ & $\cdots$ \\ \cline{2-6}
\rule{0pt}{2.5ex} $\vdots$  & $\vdots$ & $\vdots$ & $\vdots$ & $\vdots$ & $\vdots$ \\ \cline{2-6}
\rule{0pt}{2.5ex} $\glorule^n(c) $ & $\cdots$ & $\cdots$ & $a_1^{(i-n')}a_2^{(i-n')}a_3^{(i-n')}a_4^{(i-n')}$ & $\cdots$ & $\cdots$ \\ \cline{2-6}
\end{tabular}
\end{center}
    %\caption{Caption Placeholder}
\label{tab:placeholder_label}
\end{table*}

\end{example}

The following example shows that the condition about the minimality of $G$ with respect to $\glorule$ in Lemma \ref{lemma_product_simples} is necessary.

\begin{example}
Let $S$ be a simple, non-abelian group and consider the group $G=S\times S.$ 
Consider the two endomorphisms $h_{-1}$ and $h_{1}$ of $G$ defined by $h_{-1}(x,y):=(x,1),$ and $h_1(x,y)=(1,y).$ Notice that the images of these two endomorphisms commute element-wise so their product defines an homomorphism $f\in \Hom(G^3,G)$ such that $f=(h_{-1},h_0,h_1),$ where $h_0$ is the trivial endomorphism.  Consider the GCA $\glorule$ over $G$ with local rule $f$.

Notice that $\glorule$ is nothing but the product of two shift-like GCA, namely, the shift $\sigma$ and its inverse $\sigma^{-1}$. Indeed, $f((x,y),(x',y'))=(x,y')$ and, hence, if $c\in G^{\mathbb Z}$ is the configuration  such that the element in position $i$ is $c_i=(x_i,y_i)$, then $\glorule(c)$ is the configuration such that $\glorule(c)_i=(x_{i-1},y_{i+1})$.   

It is clear that, in this case, Lemma~\ref{lemma_product_simples} fails. In fact, there are no powers of $\glorule$ that are shifts. We stress that this  fact does not constitute a contradiction since $G$ is not minimal with respect to $\glorule.$ The two factors $S\times\{e\}$ and $\{e\}\times S$ are actually invariant under $\glorule.$
\end{example}

As a consequence of Lemma \ref{lemma_product_simples} we have the following result. 
\begin{theorem}\label{simple_iso_copies_}
Let $G$ be a product of $m$ finite, non-abelian, isomorphic simple groups. Let $\glorule$ be a GCA on $G$ with local rule $\locrule=(h_{-\rho},\ldots,h_{\rho})$.  Suppose that $G$ is minimal with respect to $\glorule.$ The following facts hold:\\
$(1)$ $\glorule$ is injective if and only if $\glorule$ is surjective if and only if $$\bigcap_{-\rho\leq i\leq \rho}\Ker(h_i)=\{e\};$$\\
$(2)$ $\glorule$ is topologically transitive if and only if it is surjective and $\sum_i ir_i\neq 0$, where the $r_i'$s are defined as in Lemma \ref{lemma_product_simples}; \\
$(3)$ if $\glorule$ is surjective, $\glorule$ is sensitive to initial conditions  if and only if it is topologically transitive; \\
$(4)$ if $\glorule$ is surjective, the topological entropy of $\glorule$ is $$\entropy{H}(\glorule)=\frac{\left | \sum_iir_i\right | \log(|G|)}{o},$$ where $o$ and the $o_i$'s are defined as in Lemma \ref{lemma_product_simples};\\
$(5)$ $\glorule$ is neither strongly transitive nor positively expansive. 
\end{theorem}
\begin{proof}$\ $\\
$(1)$  Any injective CA is also surjective. If $\glorule$ is surjective, since $\Z_G$ is the trivial group, $\glorule$ is also injective according to  Lemma~\ref{ker_sta__nel_centro_se_sur}. As to the equivalence between the surjectivity of $\glorule$ and the fact that the intersection of all the kernels of the $h_i$'s is trivial, see Remark~\ref{FsurKerempty}.\\
% If a GCA is surjective also the local rule $\locrule$ must be surjective. The considerations preceding this theorem imply that in this case also the converse is true. 
$(2)$  Assume that $\glorule$ is topologically transitive. Clearly, $\glorule$ is also surjective.  If, for a sake of argument,  the equality $\sum_i ir_i=0$ holds, by Lemma \ref{lemma_product_simples} it would follow  that there exists $n$ such that $\glorule^n$ is the identity. But this is impossible since $\glorule$ is topologically transitive. Hence, $\sum_i ir_i\neq 0$. As to the converse implication, if $\glorule$ is surjective and $\sum_i ir_i \neq 0$, there exists $n$ such that $\glorule^n$ is a shift-like GCA. Therefore, $\glorule$ is topologically transitive. \\
$(3)$  It is well known that any topologically transitive CA is also sensitive to initial conditions. 
Assume now that  $\glorule$ is a surjective GCA which is not  topologically transitive. By item (2), $\sum_i ir_i = 0$ and there exists $n$ such that $\glorule^n$ is the identity. Therefore, $\glorule$ is not sensitive  to initial conditions by Lemma~\ref{radius_sensitive}.\\
$(4)$  It is well known \cite[Thm. 1.2, p. 335]{robinson1995dynamical} that, if $F$ is  any continuous map over a compact topological space, it holds that $$\entropy{H}(F^k)=k\entropy{H}(F)$$ for every $k\in \mathbb N$. Moreover, it is also know that the topological entropy of the shift $\sigma^r$ on the alphabet $G$ is $|r|\log(|G|)$.
In our case, by Lemma \ref{lemma_product_simples}, we get 
\begin{align*}
o\gcd(o_1,\ldots,o_m)\entropy{H}(\glorule)&=\entropy{H}(\glorule^{o\gcd(o_1,\ldots,o_m)})\\&=\entropy{H}(\sigma^{-\gcd(o_1,\ldots,o_m)\sum_iir_i})\\&=
\gcd(o_1,\ldots,o_m)\left | \sum_iir_i\right | \log(|G|).
\end{align*}
Thus, $$\entropy{H}(\glorule)=\frac{\left | \sum_iir_i\right | \log(|G|)}{o}.$$\\
$(5)$ 
 It directly follows by Theorem~\ref{str_trans_su_nonabel_nonesiste}.
% A strongly transitive GCA is surjective but not injective so the last statement of the theorem follows from the first. 
\end{proof}

% The following result is a direct consequence of Theorem \ref{simple_iso_copies_} and of Theorem \ref{spezzatino}.

% \begin{corollary}
% Let $\glorule$ be a GCA over a finite non-solvable group $G$. Then, $\glorule$ is neither strongly transitive nor positively expansive.
% \end{corollary}

% We conjecture that in fact there are no strongly transitive GCA over non-abelian group at all.

%  \begin{conjecture}
%   Let $\glorule$ be a GCA over a non-abelian finite group $G$. Then, $\glorule$ is neither strongly transitive nor positively expansive.    
%  \end{conjecture}

\subsection{Decidability of dynamical properties}

%\todo{tutto da riguardare}

Let $C$ be a possibly infinite set of CA, such as the set of GCA.
Let $P$ be a property that a CA may or may not satisfy, such as surjectivity or topological transitivity.
$P$ is decidable for $C$ if and only if there exists an algorithm that, given a CA $F \in C$, returns ``Yes" if $F$ satisfies $P$, and ``No" otherwise. 
If, instead of a property $P$, we consider a numerical function $N: C \to \mathbb{R}$, such as topological entropy or Lyapunov exponents, $N$ is computable for $C$ if and only if there exists an algorithm that, given a CA $F \in C$, computes $N(F)$.
Deciding a property (or computing a function) involves a computational cost that the notion of decidability (or computability) does not take into account. We will therefore say that a property is efficiently decidable (or a function is efficiently computable) if the algorithm that decides whether the property holds (or computes the function) is efficient, where efficient usually means polynomial-time.
As for CA, efficient algorithms typically analyze the structure of the local rule, which is a finite object, whereas inefficient algorithms usually operate on the space-time dynamics of the CA, which is potentially infinite in size.

The literature contains a number of results related to the decidability and computability of properties and functions, respectively,  across various classes of CA. 
Below, we list some of the most significant ones. 
In what follows, we will denote by  $D$-CA the class of $D$-dimensional CA. The same notation  will also be used  for GCA.

%Unless explicitly specified, CA will be assumed to be $1$-dimensional.\\
\noindent
- Every non-trivial property of limit sets of general   1-CA is undecidable \cite{Kari94}.\\
- Surjectivity and injectivity are decidable for general 1-CA \cite{amoroso1972decision}  and undecidable for general 2-CA   \cite{kari1994reversibility}. Surjectivity and injectivity are  decidable for $D$-GCA with $D\geq 1$\cite{BeaurK24}.\\
- Topological entropy for general   1-CA is  uncomputable~\cite{Hurd_Kari_Culik_1992}. It is computable for  1-GCA on $\ZZ/m\ZZ$ and for general positively expansive 1-CA 
\cite{DamicoMM03}. 
We strongly believe that the topological entropy  is also efficiently computable for 1-GCA on abelian groups.  \\
- Strong transitivity is  efficiently decidable  for $D$-GCA on $\ZZ/m\ZZ$ with  $D\geq 1$~\cite{ManziniM99}. We strongly believe that strong transitivity is also efficiently decidable for $D$-GCA on abelian groups with $D\geq 1$  .\\
- Lyapunov exponents are efficiently computable for 1-GCA on $\ZZ/m\ZZ$~\cite{FinelliMM98}. \\
- Sentitivity to the initial conditions and topological transitivity are undecidable for general 1-CA (even when restricting to the case of reversible 1-CA) ~\cite{Lukkarila10}. \\
- Sentitivity to the initial conditions, equicontinuity, topological transitivity, ergodicity, positive expansivity, and DPO are efficiently decidable for 1-GCA on abelian groups~\cite{
DennunzioFGM2020TCS,
Dennunzio20JCSS,
DennunzioFMMP19,
DBLP:journals/access/DennunzioFM23,
DBLP:journals/isci/DennunzioFM24,
kari2000}.\\ 
- Sentitivity to the initial conditions and equicontinuity
are decidable for $D$-GCA with $D\geq 1$~\cite{BeaurK24}.\\
- Non-transitivity  is semi-decidable for $D$-GCA with $D\geq 1$~\cite{BeaurK24}.\\

The decidability results in this paper are based on the following fundamental remark.
\begin{remark}\label{decid}
Since $G$ is a finite group, the set of GCA produced by the function {\tt Explicit-Decomposition} defined in Section \ref{section_spezzatino} are computable.
\end{remark}

As a consequence of Remark~\ref{decid} and the other results presented in this paper, the following statements are true.\\

\noindent
$(1)$ If Question~\ref{conjecture_trans} had a positive answer, \textit{topological transitivity, totally transitivity, ergodicity, weakly and strongly mixing are decidable properties for  1-GCA.} By Theorem~\ref{equiv_prop}, all these properties are equivalent for 1-GCA. By  Theorems~\ref{transitive} and~\ref{simple_iso_copies_}, topological transitivity is  decidable for 1-GCA.  \\
$(2)$ \textit{If topological entropy is computable  for surjective 1-GCA on abelian groups, then it is computable for all surjective 1-GCA.} It follows by Theorem~\ref{entropy}.\\
$(3)$ \textit{If strong transitivity is decidable for 1-GCA on abelian groups then it is decidable for all 1-GCA.}
Since strongly transitive 1-GCA do not exist unless the underlying group is abelian (Theorem~\ref{str_trans_su_nonabel_nonesiste}), the decidability of strong transitivity for 1-GCA reduces to its decidability for GCA on abelian groups.\\
$(4)$ \textit{Positive expansivity is decidable for 1-GCA.}
The proof follows the same reasoning as the proof of $(3)$. Note that positively expansive CA do not exist in dimensions greater than 1.\\
%Since positively expansive GCA do not exist unless the underlying group is abelian (Theorem \ref{str_trans_su_nonabel_nonesiste}), the decidability of positive expansivity reduces to its decidability for 1-CA on abelian groups.\\
%\todo{sei sicuro che str. trans sia decidibile sugli abeliani?}
$(5)$ \textit{DPO is a decidable property for 1-GCA.} It follows by Proposition \ref{surj_open}.\\
$(6)$ \textit{Sensitivity to  initial conditions} and \textit{equicontinuity} are decidable for surjective 1-GCA (Theorem \ref{sensitivity} and Lemma \ref{equic_or_sens}).

%\todo{Luciano: da qui in poi è quasi inventato. terreno scivoloso. valutare se lasciarlo\\ Niccolò: per me va bene}

Some of the results presented in this paper are related to those proven in \cite{BeaurK24}. We believe it is useful to clarify analogies and differences between the two.

First of all, it is important to emphasize that the results obtained in \cite{BeaurK24} are based on the analysis of the space-time dynamics of the GCA under consideration. This leads to an unsustainable computational cost, even for very simple GCA. In this work, we propose a completely different approach, which relies on the analysis of the structure of the group on which the GCA is defined and on its local rule. This analysis is completely independent of the dynamics of the GCA global rule. 
Moreover, in principle, our technique could be applied to any dynamical property we desire to investigate.

On the other hand, the result obtained in \cite{BeaurK24} are in a sense more general then ours since apply to $D$-dimensional GCA defined on any subshift of $G^\Z$.

As for the computational cost of the process leading to the decidability of dynamical properties of GCA, although we have not carried out a detailed analysis, we believe it to be significantly lower than that of the algorithms proposed in \cite{BeaurK24}.

It is rather remarkable that two such distant approaches have led to a surprisingly strong convergence of results.

\section{Conclusion and further work}\label{conclusion}

%\todo{tutto da riguardare}

This paper establishes that several fundamental set-theoretic and dynamical properties of a GCA hold if and only if the same properties are satisfied by a corresponding set of GCA defined on significantly easier to study finite groups (abelian groups or products of simple non-abelian isomorphic groups). The set of such GCA is obtained by means of a novel algorithmic technique provided in this paper. 
%These groups are either abelian or products of isomorphic simple groups. 
In our opinion, our results are not only interesting in themselves, but also pave the way for tackling and solving a number of open questions related to the dynamical behavior of GCA. 

%As a consequence of our results, for example,  we prove that topological transitivity is a decidable property and that no strong transitive (and therefore no positive expansive) GCA can be defined on non-abelian groups.
% leaving open the problem of finding, if it exists, a positively expansive or at least a strongly transitive GCA on a non-abelian solvable group.
%Furthermore, we prove that the topological entropy of a surjective GCA can be computed, provided one knows how to compute the entropy of GCA defined on simple and on abelian groups—two open but much more manageable problems.

%Possible directions for further research are extending our results to other dynamical properties and to $D$ dimensional GCA with $D\geq 1$. 

Unfortunately, our results do not help to solve the problem of
explicitly establishing the relationship between the local rule defining the cellular automaton and its global dynamical behavior. 
One of the reasons why this problem is so challenging is that the local rule of a GCA cannot be expressed using an easy-to-manage algebraic formulation. In the case of abelian GCA, their study can be reduced to that of Linear CA where the local rule can be represented by a square matrix whose elements are Laurent polynomials. This allows characterizing the dynamical properties of the GCA in terms of specific properties of the characteristic polynomial of the matrix representing the local rule of the  Linear CA associated with the GCA.
As far as we know, in the case of non-abelian groups, the only practical way to represent a local rule is by describing the behavior of its associated endomorphisms, e.g., by assigning their values when applied to the generators of the group. A method to overcome this limitation would be a significant step toward an easy-to-check characterization of dynamical properties.

\begin{comment}
\paragraph{Acknowledgements}
This work was partially supported by the PRIN 2022 PNRR project ``Cellular Automata Synthesis for Cryptography Applications (CASCA)'' (P2022MPFRT) funded by the European Union – Next Generation EU, and by the HORIZON-MSCA-2022-SE-01 project 101131549 ``Application-driven Challenges for Automata Networks and Complex Systems (ACANCOS)''.
\end{comment}

\paragraph{Acknowledgements}
This work was partially supported by the PRIN 2022 PNRR project ``Cellular Automata Synthesis for Cryptography Applications (CASCA)'' (P2022MPFRT) funded by the European Union – Next Generation EU, and by the HORIZON-MSCA-2022-SE-01 project 101131549 ``Application-driven Challenges for Automata Networks and Complex Systems (ACANCOS)''.

\bibliographystyle{elsarticle-num}
\bibliography{matrix_groups_bib}

%\begin{thebibliography}{00}

%% For numbered reference style
%% \bibitem{label}
%% Text of bibliographic item

%\end{thebibliography}

%\begin{comment}
\appendix

\section{Fundamentals of groups and fields}\label{A1}

In this appendix, we recall some standard definitions and concepts from the theory of finite groups and finite fields.  

\subsection{Basic Definitions in Finite Group Theory}

\subsubsection{Groups}
A \textit{group} is a set \( G \) equipped with a binary operation \( \cdot : G \times G \to G \) (typically written as multiplication) that satisfies the following axioms:
\begin{enumerate}
    \item \textbf{Associativity}: For all \( a, b, c \in G \), we have \( (a \cdot b) \cdot c = a \cdot (b \cdot c) \).
    \item \textbf{Neutral element (identity)}: There exists an element \( e \in G \) (called the identity element) such that for all \( a \in G \), \( e \cdot a = a \cdot e = a \).
    \item \textbf{Inverse element}: For each \( a \in G \), there exists an element \( a^{-1} \in G \) (called the inverse of \( a \)) such that \( a \cdot a^{-1} = a^{-1} \cdot a = e \).
\end{enumerate}
If \( G \) is a finite set, we call \( G \) a \textit{finite group}.

A group $G$ such that $a\cdot b= b \cdot a$ for all $a,b \in G$ is said to be \textit{abelian}.

A \textit{cyclic group} is a group that is generated by a single element $g$, i.e., all its elements can be written as powers of $g$.

If $A$ and $B$ are subsets of the group $G$, $AB$ denotes the set $\{a\cdot b : a\in A;\,b\in B\}$.

\subsubsection{Homomorphisms and Isomorphisms}
A \textit{homomorphism} between two groups \( G \) and \( H \) is a map \( \varphi: G \to H \) such that for all \( a, b \in G \),
\[
\varphi(a \cdot b) = \varphi(a) \cdot \varphi(b).
\]

The set of all homomorphisms from  \( G \) to \(H\) is  denoted \( \operatorname{Hom}(G,H) \),

If the homomorphism is bijective, then \( \varphi \) is called an \textit{isomorphism}, and \( G \) and \( H \) are said to be \textit{isomorphic} (denoted \( G \cong H \)).

\subsubsection{Endomorphisms and Automorphisms}
An \textit{endomorphism} is a homomorphism from a group to itself, i.e., a homomorphism \( \varphi: G \to G \).
The set of all endomorphisms of a group \( G \) is denoted by \( \operatorname{End}(G) \).
An \textit{automorphism} is an isomorphism from a group to itself, i.e., a bijective endomorphism. The set of all automorphisms of a group \( G \), denoted \( \operatorname{Aut}(G) \), forms a group under composition of functions.

\subsubsection{Subgroups and Normal Subgroups}
A \textit{subgroup} of a group \( G \) is a subset \( H \subseteq G \) that is itself a group under the operation inherited from \( G \). That is, \( H \) must satisfy the following conditions:
\begin{enumerate}
    \item \( H \) contains the identity element of \( G \),
    \item \( H \) is closed under the group operation, i.e., for all \( h_1, h_2 \in H \), \( h_1 \cdot h_2 \in H \),
    \item \( H \) is closed under inverses, i.e., for all \( h \in H \), \( h^{-1} \in H \).
\end{enumerate}

A subgroup \( N \) of \( G \) is said to be \textit{normal} (denoted \( N \trianglelefteq G \)) if for all \( g \in G \) and \( n \in N \), we have \( g n g^{-1} \in N \). This is equivalent to saying that \( g N g^{-1} = N \) for all \( g \in G \).

\subsubsection{Kernel of a Homomorphism}
The \textit{kernel} of a homomorphism \( \varphi: G \to H \) is the set 
\[
\Ker(\varphi) = \{ g \in G : \varphi(g) = e_H \},
\]
where \( e_H \) is the identity element of \( H \). The kernel is a normal subgroup of \( G \).

\subsubsection{Center and Commutator Subgroup}
The \textit{center} of a group \( G \), denoted \( Z(G) \), is the set of elements in \( G \) that commute with every element of \( G \):
\[
Z_G = \{ g \in G : g h = h g \text{ for all } h \in G \}.
\]

The \textit{commutator} of two elements \( g, h \in G \) is defined as \( [g, h] = g h g^{-1} h^{-1} \). The \textit{commutator subgroup} (or derived subgroup) of \( G \), denoted \( G' \) or \( [G, G] \), is the subgroup generated by all commutators of elements of \( G \):
\[
G' = \langle [g, h] : g, h \in G \rangle.
\]

A group $G$ is said to be \textit{perfect} if $G=[G,G]$.

\subsubsection{Order of an Element and the Group}
The \textit{order} of an element \( g \in G \) is the smallest positive integer \( n \) such that \( g^n = e \), where \( e \) is the identity element of \( G \). If no such integer exists, \( g \) is said to have infinite order, though this is irrelevant for finite groups.

The \textit{order} of a finite group \( G \) is the number of elements in \( G \), denoted \( |G| \).

\subsubsection{Direct Product of Groups}
Given two groups \( G \) and \( H \), their \textit{direct product} is the group \( G \times H \) consisting of ordered pairs \( (g, h) \), where \( g \in G \) and \( h \in H \), with the group operation defined component-wise:
\[
(g_1, h_1) \cdot (g_2, h_2) = (g_1 \cdot g_2, h_1 \cdot h_2).
\]
The identity element of \( G \times H \) is \( (e_G, e_H) \), where \( e_G \) and \( e_H \) are the identity elements of \( G \) and \( H \), respectively.

\subsubsection{Semidirect Product}
Given two groups \( G \) and \( H \), and a homomorphism \( \varphi: H \to \operatorname{Aut}(G) \) (where \( \operatorname{Aut}(G) \) is the group of automorphisms of \( G \)), the \textit{semidirect product} of \( G \) and \( H \), denoted \( G \rtimes_\varphi H \), is the set \( G \times H \) with the operation defined by:
\[
(g_1, h_1) \cdot (g_2, h_2) = (g_1 \cdot \varphi(h_1)(g_2), h_1 \cdot h_2).
\]
If \( \varphi \) is the trivial homomorphism (i.e., \( \varphi(h) \) is the identity automorphism for all \( h \in H \)), then \( G \rtimes_\varphi H \) is the direct product \( G \times H \).

\subsubsection{Quotient Groups}
Let \( N \) be a normal subgroup of \( G \). The \textit{quotient group} \( G/N \) is defined as the set of left cosets of \( N \) in \( G \):
\[
G/N = \{ gN : g \in G \},
\]
where each coset \( gN = \{ g n : n \in N \} \). The operation on \( G/N \) is given by coset multiplication:
\[
(gN)(hN) = (gh)N \quad \text{for all } g, h \in G.
\]

The coset $gN$ is denoted by $[g]$.

% Since \( N \) is normal, the product of two cosets is well-defined. The identity element in \( G/N \) is the coset \( N \) (the coset containing the identity element of \( G \)), and the inverse of \( gN \) is \( g^{-1}N \).

The order of the quotient group \( G/N \) is given by
\[
|G/N| = \frac{|G|}{|N|}.
\]

It is customary to denote by $\pi$ the map $\pi:G\to G/N$ that associate each element $g$ to its coset $gN.$

\subsubsection{First Isomorphism Theorem}

The \textit{First Isomorphism Theorem} for groups states that if \( \varphi: G \to H \) is a homomorphism between two groups, then the quotient group \( G / \Ker(\varphi) \) is isomorphic to the image of \( \varphi \). 

More formally, the theorem asserts that the map 
\[
\psi: G / \Ker(\varphi) \to \operatorname{Im}(\varphi), \quad g \Ker(\varphi) \mapsto \varphi(g)
\]
is a well-defined isomorphism. 

For more details about group theory see \cite{Rotman_groups}.

\subsubsection{Conjugacy Classes}

The \textit{conjugacy class} of an element \( g \in G \) is the set of elements in \( G \) that are \textit{conjugate} to \( g \), in symbols
\[
\mathcal{C}(g) = \{ xgx^{-1} : x \in G \}.
\]

The size of the conjugacy class of \( g \) can be calculated using the formula
\[
|\mathcal{C}(g)| = \frac{|G|}{|C_G(g)|},
\]
where \( C_G(g) \) is the \textit{centralizer} of \( g \) in \( G \), defined as
\[
C_G(g) = \{ x \in G : xg = gx \}.
\]

More generally, the \textit{centralizer} of a subset $S\subseteq G$ is the subgroup 
\[
C_G(S) = \{ x \in G : xg = gx \; \forall g \in S\}.
\]

% \subsubsection{Burnside's Lemma}

% If a group \( G \) acts on a set \( X \), then the number of distinct orbits of the action of \( G \) on \( X \) can be determined using the following formula known as \textit{the Burnside lemma}:
% \[
% \text{Number of orbits} = \frac{1}{|G|} \sum_{g \in G} |X^g|,
% \]
% where \( X^g = \{ x \in X : g \cdot x = x \} \).

% If we consider the action of $G$ onto itself by conjugation we obtain 
% \begin{equation}\label{Burnside_conj}
% \sum_{g\in G}|C_G(g)|=m_G|G|,
% \end{equation}
% where $m_G$ is the number of conjugacy classes in $G$.

\subsubsection{Simple Groups and Quasi-Simple Groups}
A \textit{simple group} is a nontrivial group \( G \) that has no proper nontrivial normal subgroups.

A \textit{quasi-simple group} is a perfect group \( G \) such that \( G/Z(G) \) is a simple group.

\subsection{Finite Fields}

Here we recall the basic notions about finite fields. See \cite{Rotman_groups} for further details. 

A field $\mathbb K$ is is a set equipped with two operations: addition \( + \) and multiplication \( \cdot \), such that the following properties hold:

\begin{enumerate}
    \item $(\mathbb K, +)$ is a group with $0$ as identity element,
    \item $(\mathbb K\setminus{0}, \cdot )$ is a group with $1$ as identity element and
    \item multiplication distributes over addition.
\end{enumerate}

A finite field, also known as a Galois field, is a field with a finite number of elements. Let $\mathbb{F}_q$ denote a finite field with $q$ elements. It is a well-known result that the number of elements $q$ in a finite field is always a power of a prime, i.e., $q = p^m$, where $p$ is a prime number, and $m$ is a positive integer. When $m = 1$, the field $\mathbb{F}_q$ is simply the cyclic group $\ZZ_p$.

Finite fields $\mathbb{F}_q$ are unique up to isomorphism, meaning that all finite fields with $q$ elements are structurally the same. The multiplicative group of non-zero elements of $\mathbb{F}_q$, denoted by $\mathbb{F}_q^\times$, is cyclic of order $q-1$ \cite[Thm. 2.18]{Rotman_groups}.

%\EOD
%\end{comment}
\end{document}